\documentclass[11pt]{article}
\usepackage[english]{babel}
\usepackage{amscd}
\usepackage{amsmath}
\usepackage{amsthm}
\usepackage{amssymb, faktor, titlesec}
\usepackage{graphicx, tikz-cd}
\usepackage[left=2cm,right=2cm, top=2.5cm,bottom=2.5cm,bindingoffset=0cm]{geometry}
\usepackage{enumerate, multicol, makecell}
\usepackage{setspace}
\usepackage{mathabx}

\titleformat{\section}{\large\bfseries\filcenter}{\thesection}{1em}{}
\titleformat{\subsection}{\bfseries}{\thesubsection}{1em}{}
\titleformat{\subsubsection}[runin]{\bfseries}{\thesubsubsection}{1em}{}[.]
\usepackage{enumitem}
\newlist{longenum}{enumerate}{5}
\setlist[longenum,1]{label=\alph*)}

\usetikzlibrary{positioning}
\usetikzlibrary{decorations.text}
\usetikzlibrary{decorations.pathmorphing}
\usetikzlibrary{decorations.markings}

\usepackage[numbers, sort]{natbib}

\usepackage{mathtools}
\mathtoolsset{showonlyrefs}

\theoremstyle{plain}

\newtheorem{lemma}{Lemma}[section]
\newtheorem{proposition}[lemma]{Proposition}
\newtheorem{remark-definition}[lemma]{Remark-Definition}
\newtheorem{theorem}[lemma]{Theorem}
\newtheorem{corollary}[lemma]{Corollary}

\newtheorem{proposition-conjecture}[lemma]{Proposition-conjecture}

\theoremstyle{definition}

\newtheorem{definition}[lemma]{Definition}
\newtheorem{remark}[lemma]{Remark}
\newtheorem{example}[lemma]{Example}

\newcommand{\Ker}[1]{\mathrm{Ker} \, #1}

\newcommand{\rank}[1]{\mathrm{rank} \, #1}

\newcommand{\corank}[1]{\mathrm{corank} \, #1}

\newcommand{\Sing}[1]{\mathrm{Sing} \, #1}

\newcommand{\diff}[1]{\mathrm{d}  #1}

\newcommand{\diffFXi}[2]{ {\partial #1}/{\partial #2} }
\newcommand{\diffFXYp}[2]{ \frac{\mathrm{d} #1}{\mathrm{d} #2} }

\newcommand{\diffXp}[1]{ \frac{\mathrm{d} }{\diff #1} }

\newcommand{\Z}{\mathbb{Z}}

\renewcommand{\H}{\mathbb{H}}
\newcommand{\R}{\mathbb{R}}
\newcommand{\Complex}{\mathbb{C}}

\newcommand{\manifold}{{ M}}

\newcommand{\T}{\mathrm{T}}

\newcommand{\const}{\mathrm{const}}
\newcommand{\CP}{{\mathbb{C}}\mathrm{P}}

\newcommand{\Hom}{\mathrm{H}}

\newcommand{\Imm}[1]{\mathrm{Im} \, #1}
\newcommand{\Tr}[1]{\mathrm{Tr} \, #1}
\newcommand{\tr}[1]{\mathrm{tr} \, #1}

\newcommand{\eps}{\varepsilon}

\newcommand{\InnerProduct}{ \langle \, , \rangle}
\newcommand{\g}{\mathfrak{g}}

\newcommand{\so}{\mathfrak{so}}

\newcommand{\Pic}{\mathrm{Pic}}
\newcommand{\Jac}{\mathrm{Jac}}
\newcommand{\Div}{\mathrm{Div}}
\newcommand{\PDiv}{\mathrm{PDiv}}
\newcommand{\Res}{\mathrm{Res}}

\newcommand{\End}{\mathrm{End}}
\newcommand{\PGL}{\mathbb{P}\mathrm{GL}}

\newcommand{\E}{\mathrm{Id}}

\newcommand{\sP}{\mathfrak{sp}}
\newcommand{\sL}{\mathfrak{sl}}
\newcommand{\gl}{\mathfrak{gl}}
\newcommand{\un}{\mathfrak{u}}
\newcommand{\su}{\mathfrak{su}}

\newcommand{\Id}{\E}

\newcommand{\ad}{\mathrm{ad}}

\newcommand{\F}{{\pazocal H}}

\renewcommand{\deg}{\mathrm{deg}\,}
\newcommand{\mdeg}{\mathrm{m.deg}\,}

\usepackage{hyperref}

\usepackage{calrsfs}
\DeclareMathAlphabet{\pazocal}{OMS}{zplm}{m}{n}

\title{\Large Singularities of integrable systems and algebraic curves}
\author{Anton Izosimov\thanks{Department of Mathematics, University of Toronto, e-mail: \tt{izosimov@math.utoronto.ca}}}
\date{}
\begin{document}
\maketitle
\abstract{ We study the relationship between singularities of finite-dimensional integrable systems and singularities of the corresponding spectral curves.  For the large class of integrable systems on matrix polynomials,  which is a general framework for various multidimensional spinning tops, as well as Beauville systems, we prove that if the spectral curve is nodal, then all singularities on the corresponding fiber of the system are non-degenerate. 
We also show that the type of a non-degenerate singularity can be read off from the behavior of double points on the spectral curve under an appropriately defined antiholomorphic involution. 
\par
 Our analysis is based on linearization of integrable flows on generalized Jacobian varieties, as well as on the possibility to express the eigenvalues of an integrable vector field linearized at a singular point in terms of residues of certain meromorphic differentials.


}
\medskip
\tableofcontents
\medskip
\section{Introduction}
This paper is an attempt to explain the relation between two major research directions in the theory of finite-dimensional integrable systems. One direction is the study of singularities of integrable systems, and the other one is construction of explicit solutions using algebraic geometric methods. Despite the fact that both theories exist for more than thirty years, there is obviously a lack of interaction between them. 
While it is, in principle, understood that singularities occur when the algebraic curve used to derive explicit solutions (the so-called \textit{spectral curve}) fails to be smooth, almost nothing is known beyond this. In particular, it is not known what kind of singularities of integrable systems correspond to different types of curve singularities. Our paper can be regarded as a first step in this direction. 
\par
The study of singularities of integrable systems is an active research area receiving an increasing attention in the last few years. 
As is well known, the phase space of a finite-dimensional integrable system is almost everywhere foliated into invariant Lagrangian submanifolds (with the latter being tori in the compact case or open dense subsets of tori in the algebraic case). However, this description breaks down on the singular set where the first integrals become dependent. Although the set of such singular points is of measure zero, singularities are important at least for the following reasons:
\begin{itemize}
\item  The most interesting solutions (such as equilibrium points,  homoclinic and heteroclinic orbits,  stable periodic solutions, etc.) are located on singular fibers.  

\item As compared to generic trajectories, solutions lying in the singular set usually admit simpler analytic expressions (theta functions of smaller genus or even elementary functions). Roughly speaking, singularities play the same role in finite-dimensional integrable systems as do finite-gap and multisoliton solutions in integrable PDEs.


\item Many analytic effects (such as Hamiltonian monodromy) are determined by the structure of singular fibers. 
\item In the case of an algebraically integrable system, singularities  are related to degenerations of Abelian varieties.
\end{itemize}
Local properties of singularities of integrable systems were studied, in particular, by J.\,Vey~\cite{vey}, H.\,Ito~\cite{ito}, L.H.\,Eliasson~\cite{Eliasson}, and E.\,Miranda and N.T.\,Zung~\cite{miranda2004equivariant}.  Besides, there has been a lot of activity in the study of topology of singular fibers (see, e.g., the review~\cite{BolOsh} and references therein). As a result, by now there exists quite an accomplished theory that classifies main types of singularities for integrable systems. However, there exist almost no efficient techniques for describing singularities of particular systems, and except for several examples (see, for instance, recent works~\cite{babelon2015higher, bouloc2015singular, ratiu2015u}), singularities of multidimensional systems have not been understood. \par
In the present paper, we study singularities of integrable systems from the point of view of algebraic geometry. 
 The algebraic geometric approach to integrable systems is based on the observation that most interesting examples of such systems can be written as \textit{Lax equations with a spectral parameter}, that are equations of the form
\begin{align*}
\diffXp{t} L( \lambda) = [L( \lambda), A( \lambda)],
\end{align*}
where $L(\lambda)$ and $A(\lambda)$ are elements of the loop algebra of a matrix Lie algebra, or, more generally, matrix-valued meromorphic functions on $\CP^1$ or another Riemann surface. 
 The advantage of the Lax form is that it enables one to apply algebraic geometric methods for proving complete integrability and deriving explicit solutions. The key idea of this approach is to consider the curve
  $
\{ (\lambda, \mu) \in \Complex^2 \mid \det(L(\lambda) - \mu \E) = 0\},
$
 known as the \textit{spectral curve}. Then a classical result says that under certain additional assumptions the Lax equation linearizes on the Jacobian of the spectral curve. 
%
%
%
%
%
%
%
%
 For further details, see, e.g., the reviews~\cite{DKN, DMN, reyman1994group, hitchin1999riemann}.

{
}
\par
Nowadays, the algebraic geometric framework based on the notion of a Lax representation is considered as the most standard approach to integrability. This method has proved to be very powerful not only for constructing new examples and explicit integration, but also for studying topological properties of integrable systems (see, in particular, \cite{Audin, audin1993varietes}).
The goal of the present work is to show that singularities also fit into this algebraic geometric picture in the most natural way. 
\par
A well-known principle in the theory of integrable systems says that \textit{if the spectral curve is smooth, then so is the corresponding fiber of the system}. Accordingly, if one aims to study singularities of an integrable system, then it is necessary to consider {singular spectral curves}. It is natural to conjecture that spectral curves with generic, i.e., nodal, singularities correspond to non-degenerate singularities of the system. Non-degenerate singular points of integrable systems are analogous to Morse critical points of smooth functions. In particular, they are stable under small perturbations and are linearizable  in the sense that  the Lagrangian fibration near a non-degenerate singular point is symplectomorphic to the one given by quadratic parts of the integrals. Furthermore, any non-degenerate singularity can be locally represented as the product of standard singularities of three possible types: elliptic, hyperbolic and focus-focus.  The complete local
invariant of such a singularity is the so-called Williamson type of the point, a triple  $(k_e, k_h, k_f)$ of non-negative integers being the numbers of  elliptic, hyperbolic and focus-focus components in this decomposition (see Section~\ref{sec:isnds} for precise definitions). Also note that the notion of type is only defined in the real case. For complex-analytic systems, there is essentially only one type of non-degenerate singularities (similarly to how Morse critical points of different indices become the same after complexification).\par
The main result of the present paper, which we make precise in Theorems~\ref{thm1} and~\ref{thm2} below, is the following
\begin{theorem}\label{thm0}
For polynomial matrix systems, nodal spectral curves correspond to non-degenerate singular fibers. In the real case, the type of a non-degenerate singularity is determined by the behavior of double points of the spectral curve under the antiholomorphic involution coming from the real structure of the system.
\end{theorem}
%
Polynomial matrix systems are natural integrable systems on finite-dimensional coadjoint orbits in the loop algebra of $\gl_n$.
These systems were studied, among others, by P.\,Van Moerbeke and D.\,Mumford~\cite{MVM}, A.G.\,Reyman and M.A.\,Semenov-Tian-Shansky~\cite{Reyman1,Reyman2}, M.\,Adler and P.\,Van Moerbeke~\cite{adler2, Adler},  M.R.\,Adams, J.\,Harnad, and J.\,Hurtubise~\cite{adams1990}, A.\,Beauville~\cite{beauville1990jacobiennes}, and L.\,Gavrilov~\cite{Gavrilov}. 
Note that polynomial matrix systems can be regarded as some kind of universal integrable hierarchies, and many familiar systems, including the geodesic flow on an ellipsoid, the Euler, Lagrange and Kovalevskaya tops, as well as their various multidimensional generalizations, can be obtained as restrictions of polynomial matrix systems to appropriate subspaces.
\begin{example}\label{ex0}
The Euler-Manakov top~\cite{Manakov} can be obtained as a restriction of a polynomial matrix system to the twisted loop algebra associated with the Cartan involution of $\gl_n$. The Lagrange top is per se an example of a polynomial matrix system, see Section \ref{sec:lagrange} below. The Kovalevskaya top can be obtained as a restriction of a polynomial matrix system to the twisted loop algebra associated with the Cartan involution of $\so_{3,2}$, as explained in~\cite{bobenko1989kowalewski}.
\end{example}
\begin{example}
Beauville systems \cite{beauville1990jacobiennes} are obtained from polynomial matrix systems by means of Hamiltonian reduction. For this reason, Theorem \ref{thm0} holds for Beauville systems as well.
\end{example}
There is an evidence that a statement similar to that of Theorem \ref{thm0} is true for even more general integrable systems. We confine ourselves to polynomial matrix systems solely for the purpose of self-consistency and conciseness of the exposition. At the same time, our {techniques} are applicable for a large variety of systems. 
Note that for more general integrable systems, Theorem~\ref{thm0} should be understood rather as a guiding principle than a precise statement (similarly to the above-mentioned principle saying that the smoothness of the spectral curve implies the smoothness of the corresponding fiber). In particular, in many examples the spectral curve carries \textit{several} antiholomorphic involutions, and thus the statement of Theorem~\ref{thm0} should be refined to take all these involutions into account. For instance, this is always so for systems which linearize on a Prym variety. This situation will be discussed in detail in a separate publication (joint with K.\,Aleshkin, in preparation).


\par

\par
Our proof of Theorem~\ref{thm0} for polynomial matrix systems is based on linearization of Lax equations on generalized Jacobian varieties (see Lemma~\ref{lemma1} below), as well as on the possibility to express the eigenvalues of a Lax flow linearized at a singular point in terms of residues of appropriately defined meromorphic differentials (Lemma~\ref{lemma2}). This technique for studying singularities via generalized Jacobians develops in particular the approach of M.\,Audin and her collaborators \cite{audin1993varietes, Audin, audin2002actions}, L.\,Gavrilov~\cite{Gavrilov2}, and R.~Inoue, T. Yamazaki, and P. Vanhaecke~\cite{inoue}.



We illustrate Theorem~\ref{thm0} with two concrete examples. The first one (shift of argument systems, Section~\ref{sec:sas}) generalizes recent results of T.\,Ratiu and D.\,Tarama~\cite{ratiu2015u}. The second example is the Lagrange top (Section~\ref{sec:lagrange}). In this case, we use Theorem~\ref{thm0} to give a geometric interpretation of classical results.\par

We note that this paper deals with local properties of singularities (i.e., non-degeneracy and Williamson type) only. Some applications of algebraic geometry to the semi-local singularity theory (i.e., description of singular fibers) can be found in~\cite{inoue,gothen2012singular, LaxStab, izosimov2015matrix}. More degenerate, i.e., non-nodal, singular spectral curves are also beyond the scope of the present paper. Such curves should, in general, correspond to degenerate singularities. Note that in contrast to singularities of smooth functions and maps, there exists no classification of degenerate singularities of integrable systems (and such a theory is seemingly beyond reach). However, we conjecture that for algebraic geometric integrable systems (i.e., systems which admit a Lax representation and linearize on the Jacobian of the spectral curve) degenerate singularities can be described in terms of singularities of plane curves (which are well-understood). A construction of such a correspondence between curve singularities  and singularities of integrable systems is, in our opinion, a very interesting and important problem.\par
%
\par
\medskip
{\bf Acknowledgments.} This work was partially supported by the Dynasty Foundation Scholarship and an NSERC research grant. The author is grateful to Konstantin Aleshkin, Alexey Bolsinov, Boris Khesin, and Askold Khovanskii for fruitful discussions and useful remarks. 
\bigskip \section{Main definitions}\label{sec:defs}
\subsection{Integrable systems, non-degenerate singularities, and real forms}\label{sec:isnds}
In this section, we give a formal definition of integrability and also discuss singularities and real forms of integrable systems. Note that although all systems considered in the paper are defined on Poisson manifolds, it will be convenient for us to introduce all definitions in the symplectic setting. Since every Poisson manifold is foliated into symplectic leaves, these definitions extend to the Poisson case in a natural way.

Let $\manifold$ be a real analytic or complex analytic manifold endowed with an analytic symplectic structure~$\omega$. Denote the space of analytic functions on $\manifold$ by $\pazocal O(\manifold)$.
This space is a Lie algebra with respect to the Poisson bracket associated with the symplectic form $\omega$.

\begin{definition}\label{def:integrability} Let $M$ be a real analytic or complex analytic symplectic manifold of dimension $2n$.
\begin{enumerate} \item
An \textit{integrable system} $\F$ on $\manifold$ is a commutative Lie subalgebra of $\pazocal O(\manifold)$ such that the space
	 $\diff \F(x) := \{ \diff H(x) \mid H \in \F\} \subset \T^*_x {\manifold}$ has the maximal possible dimension $n$ for almost every $x \in \manifold$.
	 \item A point $x \in \manifold$ is called \textit{singular} for an integrable system $\F$ on $\manifold$ if $\dim \diff \F(x)   < n$. The number $ \dim \diff \F(x)$ is called the \textit{rank} of the  singular point $x$. The number $n - \dim \diff \F(x)$ is called the \textit{corank} of $x$.
	 \item A \textit{fiber} of an integrable system $\F$ is a joint level set of all functions $H \in \F$. A fiber is called \textit{singular} if it contains at least one singular point, and \textit{non-singular} otherwise.
	 \end{enumerate}
\end{definition}
By the classical Arnold-Liouville theorem, each non-singular fiber of an integrable system $\F$ is a Lagrangian submanifold of the ambient manifold $M$ (moreover, any connected component of a compact non-singular fiber is a Lagrangian torus). Singular fibers are, on the contrary, in general not manifolds, however their structure is still not too complicated, provided that all singularities on the fiber satisfy the \textit{non-degeneracy} condition that we introduce below.
\par
%
Let $\F$ be an integrable system on a symplectic manifold $\manifold$. For any $H \in \F$, let $\mathrm X_H = \omega^{-1}\diff H$ be the corresponding Hamiltonian vector field.
Assume that a point $x \in \manifold$ is singular for the system~$\F$, and let $H \in \F$ be such that $\diff H(x) = 0$ and thus $\mathrm X_H = 0$. Take such $H$, and consider the linearization of the vector field $\mathrm{X}_H$ at the point $x$. This is a linear operator $\mathrm D\mathrm{X}_{H}: \T_x \manifold\to \T_x \manifold$.
Note that since the vector field $\mathrm X_H$ is Hamiltonian, we have $\mathrm D\mathrm{X}_{H} \in \sP(\T_x \manifold,\, \omega)$, where $$ \sP(V,\, \omega) := \{ D \in \End(V) \mid \omega(D\xi, \eta) + \omega(\xi, D\eta) = 0 \, \forall \, \xi, \eta \in V \}$$ is the symplectic Lie algebra. Furthermore, since for any $H_1, H_2 \in \F$ we have $[\mathrm X_{H_1}, \mathrm X_{H_2}] = 0$, it follows that $[\mathrm D\mathrm X_{H_1}, \mathrm D\mathrm X_{H_2}] = 0$, i.e. $\{ \mathrm{DX}_{H} \mid H \in \F; \diff H(x) = 0\}$ is a commutative subalgebra of $\sP(\T_x \manifold,\, \omega)$.
\par
Now, consider the space $\pazocal X_{\F} := \{\mathrm X_H(x) \mid H \in \F\} \subset \T_x \manifold$, and let $\pazocal X_{\F}^\bot  \subset \T_x \manifold$ be its orthogonal complement with respect to $\omega$.  Then, using once again the commutativity of $\F$, we get that $\pazocal X_{\F} \subset \pazocal X_{\F}^\bot $ (i.e., the space $\pazocal X_{\F}$ is isotropic), so the form $\omega$ descends to the quotient $\pazocal X_{\F}^\bot \,/\, \pazocal X_{\F} $, turning the latter into a symplectic space. Furthermore, all operators $\mathrm{DX}_{H} $, where $H \in \F$ and $\diff H(x) = 0$, also descend to  $\pazocal X_{\F}^\bot \,/\, \pazocal X_{\F}$. Denote the collection of these descended operators by $\pazocal{D}_{\F}^{} $.
This is, by construction, a commutative subalgebra of $ \sP(\pazocal X_{\F}^\bot \,/\, \pazocal X_{\F}, \,\omega)$.

\begin{definition}\label{nd}
Let $x \in \manifold$ be a singular point of an integrable system $\F$. Then $x$ is called \textit{non-degenerate} if the associated subalgebra $ \pazocal{D}_{\F}^{} \subset  \sP(\pazocal X_{\F}^\bot \,/\, \pazocal X_{\F}, \,\omega)$ is a Cartan subalgebra.
\end{definition}
It turns out that the classification of non-degenerate singular points for integrable systems coincides with the classification of Cartan subalgebras in the symplectic Lie algebra, considered up to conjugation. In particular, in the complex case all Cartan subalgebras are conjugate to each other, and there is essentially only one type of non-degenerate singularities.
\begin{theorem}[J.\,Vey~\cite{vey}]
\label{EliassonThmComplex}
Let $\F$ be a holomorphic integrable system on a complex symplectic manifold $\manifold$, and let $x \in \manifold$ be its non-degenerate singular point of rank $r$. Then there exists a Darboux chart\footnote{In what follows, we say that $(p_1, q_1,\dots, p_n, q_n)$ is a Darboux chart if the symplectic form $\omega$ written in this chart has the canonical form $\omega = \diff p_1 \wedge \diff q_1 + \dots + \diff p_n \wedge \diff q_n$.}  $(p_1, q_1,\dots, p_n, q_n)$ centered at $x$ (i.e., $p_i(x) = q_i(x) = 0$)  such that each $H \in \F$ can be written as $H = H(f_1, \dots, f_n)$, where 
\begin{align*}
f_i = \left[\begin{aligned}
&p_i \quad &\mbox{for }  i \leq r,\\
&p_iq_i \quad &\mbox{for } i > r.\\
\end{aligned}\right.
\end{align*}
Moreover, there exist functions $H_1, \dots, H_n \in \F$ such that $\det\left(\diffFXi{H_i}{f_j}(0)\right) \neq 0$.
\end{theorem}
From the geometric point of view, Theorem~\ref{EliassonThmComplex} means that near a non-degenerate singular point the Lagrangian fibration of $\manifold$ by joint level sets of functions $H \in \F$ is locally symplectomorphic to the direct product of a non-singular fibration and several copies of the \textit{nodal} fibration defined by the function $pq$ in the neighborhood of the origin $0 \in (\Complex^2, \diff p \wedge \diff q)$.\par
 In the real case, not all Cartan subalgebras are conjugate to other, but there are finitely many conjugacy classes described as follows. Assume that $\mathfrak t\subset \sP_{2m}(\R)$ is a Cartan subalgebra, and let $\xi \in \mathfrak t$ be a generic element. Then the eigenvalues of $\xi$ (in the fundamental representation) have the form
\begin{align*}
\pm \alpha_{1}\sqrt{-1}, \dots, \pm \alpha_{k_e}\sqrt{-1}\,,\quad \pm \beta_{1}, \dots, \pm \beta_{k_h}, \quad \pm \gamma_{1} \pm \delta_{1}\sqrt{-1}, \dots, \pm \gamma_{k_f} \pm \delta_{k_f}\sqrt{-1}\,,
\end{align*}
where $\alpha_i, \beta_i, \gamma_i, \delta_i \in \R$, and $k_e$, $k_h$, $k_f$ are non-negative integers satisfying $k_e + k_h + 2k_f = m$. The triple $(k_e,k_h,k_f)$ is the same for any regular $\xi \in \mathfrak t$ and is called the \textit{type} of the Cartan subalgebra $\mathfrak t$. According to J.\,Williamson~\cite{Williamson}, two Cartan subalgebras of $ \sP_{2m}(\R)$ are conjugate to each other if and only if they are of the same type.
\begin{definition}
	The Williamson \textit{type} of a non-degenerate point $x$ is the type of the associated Cartan subalgebra $\pazocal{D}_{\F}^{} \subset \sP(\pazocal X_{\F}^\bot\, /\, \pazocal X_{\F},\, \omega)$.
\end{definition}
Note that for every non-degenerate point $x$ of corank $k$ and type $(k_e,k_h,k_f)$, one has
$$
 k_e + k_h + 2k_f  = k\,.
$$ 
The numbers $k_e,k_h,k_f$ are called the numbers of \textit{elliptic}, \textit{hyperbolic}, and \textit{focus-focus} components respectively. %
%
\par
The following result is a real-analytic version of Theorem~\ref{EliassonThmComplex}.
\begin{theorem}
\label{EliassonThm}
	Let $\F$ be a real analytic integrable system on a symplectic manifold $\manifold$,
	 and let $x \in \manifold$ be its non-degenerate singular point of rank $r$ and type $(k_e,k_h,k_f)$. Then there exists a Darboux chart $(p_1, q_1,\dots, p_n, q_n)$ centered at $x$ such that each $H \in \F$ can be written as $
	H = H(f_1, \dots, f_n)
	$, where 
\begin{align*}
f_i = \left[\begin{aligned}
&p_i \quad &\mbox{for } i \leq r,\\
&p_i^2 + q_i^2 \quad &\mbox{for } r < i \leq r+k_e,\\
&p_iq_i \quad &\mbox{for }  r+ k_e < i \leq r + k_e + k_h,\\
&p_iq_i + p_{i+1}q_{i+1} \quad &\mbox{for } i = r+ k_e+k_h+ 2j -1, 1 \leq j \leq k_f,\\
&p_{i-1}q_i - p_{i}q_{i-1}  \quad &\mbox{for } i = r+ k_e+k_h+2j,  1 \leq j \leq k_f. \\\end{aligned}\right.
\end{align*}
Moreover, there exist functions $H_1, \dots, H_n \in \F$ such that $\det\left(\diffFXi{H_i}{f_j}(0)\right) \neq 0$.
\end{theorem}
 From Theorem~\ref{EliassonThm} it follows, in particular, that near a non-degenerate singular point, the singular Lagrangian fibration associated with $\F$ is locally symplectomorphic to a direct product of the following standard fibrations:
 \begin{enumerate}
    \item a non-singular fibration given by the function $p$ in the neighborhood of the origin in $(\R^2, \diff p \wedge \diff q)$;
 \item an \textit{elliptic fibration} given by the function $p^2 + q^2$ in the neighborhood of the origin in $(\R^2, \diff p \wedge \diff q)$;
  \item a \textit{hyperbolic fibration} given by the function $pq$ in the neighborhood of the origin in $(\R^2, \diff p \wedge \diff q)$;
       \item a \textit{focus-focus (or nodal) fibration} given by the commuting functions  $p_1q_1 + p_2q_2$, $p_1q_2 - q_1p_2$ in the neighborhood of the origin in  $(\R^4, \diff p_1 \wedge \diff q_1 + \diff p_2 \wedge \diff q_2)$.
    
 \end{enumerate}
In the $C^\infty$-category, this decomposition result is known as Eliasson's theorem \cite{Eliasson, miranda2004equivariant}.
 \par \smallskip
Now, we discuss the relation between singularities of real integrable systems and singularities of their complexifications. Let $(\manifold, \omega)$ be a complex symplectic manifold. Assume that $\manifold$ is endowed with a real structure, that is an antiholomorphic involution $\tau$ preserving the symplectic form (i.e., $\tau^*\omega = \bar \omega$). Then the real part $\manifold_\R = \{ x \in \manifold \mid \tau(x) =x\}$ of $\manifold$ has a natural structure of a real analytic symplectic manifold. Furthermore, if $\F \subset \pazocal F(\manifold)$ is an integrable system which is invariant with respect to the action of $\tau$ on holomorphic functions (i.e., $\tau^*\F = \bar \F$), then $\F_\R = \{ H \in \F \mid \tau^*H = \bar H\}$ is a real analytic integrable system on $\manifold_\R$.
\begin{definition}\label{realFormIS}
The system $\F_\R$ is called a \textit{real form} of the system $\F$.
\end{definition}
Clearly, one has the following result.
\begin{proposition}\label{complexification}
Singularities of $\F_\R$ are exactly those singularities of $\F$ which belong to $\manifold_\R$. Furthermore, the rank of a point $x \in \manifold_\R$ for the system $\F_\R$ is the same as for the system $\F$, and any singular point $x \in \manifold_\R$ is non-degenerate for $\F_\R$ if and only if it is non-degenerate for $\F$.
\end{proposition}
 This proposition enables one to study singularities of real integrable systems by considering their complexifications. The only information which cannot be obtained in this way is the type of a singularity.

\medskip
\subsection{Polynomial matrix systems}\label{sec:ismp}
A \textit{polynomial matrix system} is a universal integrable system on the space of fixed degree matrix polynomial with a constant leading term. The corresponding commuting functions are spectral invariants of a matrix polynomial, while the associated Hamiltonian flows are given by Lax equations. The present section is devoted to the formal definition of these integrable systems.\par
Fix positive integers $d, n$ and a matrix $J \in \gl_n$ with distinct eigenvalues (in what follows, $\gl_n$ stands for $\gl_n(\Complex)$ unless otherwise specified). Consider the affine space
 \begin{equation*}
\pazocal P_d^J (\gl_n):= \left\{ \sum\nolimits_{i=0}^{d} L_i\lambda^i \mid  L_i \in \gl_n,\,L_d = J \right\} 
\end{equation*} 
of $\gl_n$-valued degree $d$ polynomials with leading coefficient $J$. This space has a Poisson structure which is defined as follows. 
Identify the tangent space to $\pazocal P_d^J (\gl_n)$ at every point $L$ with the set of matrix polynomials of the form
\begin{align}
\label{tangent}
X(\lambda) = \sum\nolimits_{i=0}^{d-1} X_i\lambda^i, \quad  X_i \in \gl_n.
\end{align}
Then the cotangent space  to $\pazocal P_d^J (\gl_n)$ at $L$ can be identified with the set of matrix Laurent polynomials of the form
\begin{align}
\label{cotangent}
Y(\lambda) = \sum\nolimits_{i=1}^{d} Y_i\lambda^{-i}, \quad  Y_i \in \gl_n
\end{align}
with the pairing between tangent and cotangent vectors given by
\begin{align}
\label{ResTrForm}
\langle X(\lambda), Y(\lambda) \rangle := \Res_{\lambda = 0}\left(\Tr X(\lambda) Y(\lambda)\right)\diff \lambda\,.
\end{align}
Now, for any holomorphic functions $H_1, H_2 \in \pazocal O(\pazocal P_d^J (\gl_n))$, let
\begin{align}
\label{PoissonI}
\{ H_1, H_2 \}(L) := \langle L, [\diff H_1(L), \diff H_2(L)] \rangle\,,
\end{align}
where the pairing $\InnerProduct$ in the latter formula is given by \eqref{ResTrForm}. For the reader familiar with the $r$-matrix formalism, we also give an $r$-matrix form of this Poisson bracket:
\begin{align}\label{Poisson2}
\{ L(\lambda), L(\mu)\} = [r(\lambda - \mu), L(\lambda) \otimes \Id + \Id \otimes L(\mu)]\,,
\end{align}
where $$
r(z):= -\frac{1}{z}\sum\nolimits_{ij} E_{ij} \otimes E_{ji}\,.
$$
Here $E_{ij}$ are the canonical basis matrices, and
$$
\{ L(\lambda), L(\mu)\} := \sum\nolimits_{ijkl} \{ L_{ij}(\lambda), L_{kl}(\mu)\}E_{ij} \otimes E_{kl}\,.
$$

The bracket defined by formulas~\eqref{PoissonI} and \eqref{Poisson2} has the property that the spectral invariants of $L(\lambda)$ Poisson-commute, and that the corresponding commuting flows have the Lax form. More specifically, one has the following version of the Adler-Kostant-Symes theorem (see, for example, Appendix 2 of \cite{Audin} for a general statement of this theorem).
\begin{theorem}[see, e.g., Chapter 3.3 of \cite{babelon}]
Let $\psi(\lambda, \mu) \in \Complex[\lambda^{\pm 1}, \mu]$ be a function which is a Laurent polynomial in $\lambda$ and polynomial in $\mu$. Define a holomorphic function
$H_\psi \colon \pazocal P_d^J (\gl_n) \to \Complex$  by
\begin{align}
\label{HamI}
H_{\psi}(L):= \Res_{\lambda = 0} \left(\Tr\psi(\lambda, L(\lambda))\right)\diff \lambda\,.
\end{align}
Then:
\begin{enumerate} \item The functions $H_\psi$ are mutually in involution with respect to the Poisson bracket~\eqref{PoissonI}. The collection of all such functions is an integrable system.
\item For any $\psi$, the Hamiltonian flow defined by the function $H_\psi$ has a Lax form
\begin{align}\label{loopI}
\diffXp{t}L(\lambda) = [L(\lambda), \phi(\lambda, L(\lambda))_+]\,,
\end{align}
where $\phi = \partial \psi / \partial \mu$, and the subscript $+$ denotes taking the polynomial in $\lambda$ part.
\end{enumerate}
\end{theorem}
\begin{definition}
The system $\F = \{ H_\psi \mid  \psi \in \Complex[\lambda^{\pm 1}, \mu]\} \subset \pazocal O(\pazocal P_d^J (\gl_n))$, where $H_\psi$ is given by~\eqref{HamI}, is called a \textit{polynomial matrix system}.
\end{definition}
Fibers of the polynomial matrix system are, by definition, joint level sets of $H_\psi$'s. However, instead of fixing the values of $H_\psi$'s, it is more convenient to fix the
 \textit{spectral curve}.
\begin{definition}
The  \textit{spectral curve} $C_L$ associated with the matrix polynomial $L$ is a plane algebraic curve given by
$
C_L :=  \{ (\lambda, \mu) \in \Complex^2 \mid \det(L(\lambda) - \mu \E) = 0\}
$. 
\end{definition}
Note that since the Hamiltonians $H_\psi$ and the coefficients of the characteristic polynomial $\det(L(\lambda) - \mu \E)$ are expressible in terms of each other by means of Newton's identities, integrable systems defined by these two families of functions are, essentially, the same. In particular, the spectral curve is preserved by flows \eqref{loopI}, and there is a one-to-one correspondence between  fibers of the polynomial matrix system and spectral curves. 
\begin{remark}\label{multiHam}
Note that there also exist other Poisson structures on $\pazocal P_d^J (\gl_n)$  turning $\F$ into an integrable system, i.e., polynomial matrix systems are multi-Hamiltonian (see, in particular, \cite{semenov}). The reason why we choose the bracket \eqref{PoissonI} is that this bracket has a nice property of being constant rank (see Proposition~\ref{leaves} below). At the same time, all our results are valid for other Poisson structures as well, see Remark~\ref{manyPoisson}.
\end{remark}

\begin{remark}
Matrix polynomial systems on  the space $\pazocal P_d^J (\sL_n) := \{ L(\lambda) \in \pazocal P_d^J (\gl_n) \mid \tr L \equiv 0\}$ are defined in the same way. Furthermore, since the embedding $\pazocal P_d^J (\sL_n) \hookrightarrow \pazocal P_d^J (\gl_n)$ is Poisson, all the results for the $\gl$ case also hold in the $\sL$ case.
\end{remark}
\begin{remark}\label{ex:asm}
For $d = 1$, the space  $\pazocal P_d^J (\gl_n)$ can be naturally identified with the Lie algebra $\gl_n$ itself. Under this identification, the Poisson bracket given by~\eqref{PoissonI} becomes the so-called \textit{frozen argument bracket} $$\{H_1,H_2\}(X) := \Tr J[\diff H_1(X), \diff H_2(X)]\,$$ (here we identify $\gl_n$ with $\gl_n^*$ be means of the $\ad$-invariant form $\Tr XY$), and the corresponding polynomial matrix system coincides with the \textit{shift of argument} system~\cite{MF}.
\end{remark}
\medskip
\subsection{Real forms of polynomial matrix systems}\label{sec:rpms}
Real forms of polynomial matrix systems are obtained by replacing the space  $\pazocal P_d^J (\gl_n)$ in the above construction with its real analog $\pazocal P_d^J (\g) := \pazocal P_d^J (\gl_n) \cap \g \otimes \Complex[\lambda] $, where $\g$ is a real form of $\gl_n$ (obviously, for the space $ \pazocal P_d^J (\g)$ to be non-empty, one should have $J \in \g$). Real forms of $\gl_n$ are listed in Table \ref{rfgnc}. The second column shows the antilinear involution of $\gl_n$ whose fixed points set is the given real form ($\E_k$ stands for the $k \times k$ identity matrix).
\begin{table}[t]
\centering
\begin{tabular}{|c|c|}\hline $\vphantom{\displaystyle \sum}$Real form & Defining involution \\\hline   $\vphantom{\displaystyle \sum}\gl_n(\R)$ & $\vphantom{ \int}X \mapsto \bar X$ \\\hline $\vphantom{\displaystyle \int\limits_1^1}$ \begin{tabular}{c}$\vphantom{\displaystyle \sum}\un_{k,l}$\\ $\vphantom{\displaystyle \sum}k + l =n$ \end{tabular}& \begin{tabular}{c}$\vphantom{\displaystyle \int} X \mapsto -I \bar X^t I^{-1}$  \\ 
$\vphantom{\displaystyle \int\limits_1}I = \left(\begin{array}{cc}\E_k & 0 \\0 & -\E_l\end{array}\right)$
 \end{tabular}\def\arraystretch{1.5} \\\hline  $\vphantom{\displaystyle \int\limits_1^1}$\begin{tabular}{c}$\vphantom{\displaystyle \sum}\gl_k(\H)$ \\ $\vphantom{\displaystyle \sum}2k = n$\end{tabular} &  \begin{tabular}{c}$\vphantom{\displaystyle \int} X \mapsto \Omega \bar X \Omega^{-1}$\\$\vphantom{\displaystyle \int\limits_1}\Omega = \left(\begin{array}{cc}0 & -\E_k \\\E_k & 0\end{array}\right)$ \end{tabular} \\\hline \end{tabular}
\caption{Real forms of $\gl_n(\Complex)$.}\label{rfgnc}
\end{table}
%
 
 \begin{definition}
 We will call the algebras $\gl_n(\R)$ and $\gl_{k}(\H)$ \textit{real forms of non-compact type}. The algebras $\un_{k,l}$ will be called \textit{real forms of compact type}.
 \end{definition}
 \begin{remark}
 These two types of real forms are distinguished by the behavior of the spectrum in the fundamental representation. For real forms of non-compact type, the spectrum is symmetric with respect to the real axis, while in the compact case the spectrum is symmetric with respect to the imaginary axis. Note that out of real forms of compact type, only the Lie algebra $\un_n$ is actually compact.\end{remark}
Let $\g$ be a real form of $\gl_n$. Define an antiholomorphic involution $\tau \colon \Complex^2 \to \Complex^2$ by 
\begin{align}\label{involution}
\begin{aligned}
(\lambda, \mu) &\mapsto (\bar \lambda, \bar \mu) \quad \mbox{ for } \g \mbox{ of non-compact type},\\
(\lambda, \mu) &\mapsto (\bar \lambda, -\bar \mu) \quad \mbox{ for } \g \mbox{ of compact type}.
\end{aligned}
\end{align}
\newpage
\begin{proposition}\label{propInv} For any real form $\g$ of $\gl_n$, one has the following:
\begin{enumerate}
\item The set of $ H_\psi$'s with $\psi$ satisfying $\tau^*\psi = \bar \psi$, considered as functions on $\pazocal P_d^J (\g)$, is a real form (in the sense of Definition \ref{realFormIS}) of the polynomial matrix system on $\pazocal P_d^J (\gl_n)$ .
\item The corresponding Hamiltonian flows have the form \eqref{loopI} where $\phi$ satisfies 
\begin{align}
\label{invflows}
\tau^* \phi = \begin{cases} \bar \phi \quad \mbox{ for } \g \mbox{ of non-compact type},\\
-\bar \phi \quad \mbox{ for } \g \mbox{ of compact type}.
\end{cases}
\end{align}
\item For any $L \in \pazocal P_d^J (\g)$, the spectral curve $C_L$ is invariant under the action of $\tau$.
\end{enumerate}
\end{proposition}
\begin{example}
The polynomial matrix system on $\pazocal P_2^J (\su_2)$ is the classical Lagrange top (see Section~\ref{sec:lagrange} below for details). The function $\phi$ corresponding to the physical Hamiltonian of the top reads
\begin{align*}
\phi(\lambda, \mu) = \alpha \sqrt{-1}\lambda + \beta \lambda^{-1}\mu + \gamma \sqrt{-1} \lambda^{-3}\mu^2\,,
\end{align*}
where $\alpha, \beta, \gamma \in \R$ are real constants. 
\end{example}
\begin{remark}
Note that most interesting examples of integrable systems are obtained from the polynomial matrix system by intersecting the phase space of the latter (which can be naturally viewed as an affine subspace of the loop algebra $\gl_n \otimes \Complex[\lambda, \lambda^{-1}]$) with a (possibly twisted) loop algebra of another matrix Lie algebra (see, in particular, Example \ref{ex0}). Note that although these systems embed into the polynomial matrix system, this embedding is not Poisson. For this reason, the results of the present paper do not directly apply to this kind of restricted systems. However, as we plan to explain in our forthcoming publication, this difficulty can be quite easily overcome. 
\end{remark}
\bigskip

\section{Main results}
\subsection{Rank and non-degeneracy of singular points}\label{sec:rnd}
In this section, we discuss non-degenerate singularities of complex polynomial matrix systems (see Section \ref{sec:ismp} for the definition of these systems). Real forms of polynomial matrix systems and Williamson types of the corresponding singularities are discussed in Section~\ref{sec:type} below.
%
\begin{definition}\label{def:nc1}
A plane affine algebraic curve $C$ given by the equation $f= 0$ is said to have an \textit{(ordinary) double point} (equivalently, a \textit{node}), at $P \in C$ if $\diff f(P) = 0$, and $\det(\diff^2f)(P) \neq 0$. A curve is called \textit{nodal} if all its singularities are ordinary double points\footnote{Note that according to this definition smooth curves are a particular case of nodal curves.
}. 
\end{definition}

 Let $L \in \pazocal P_d^J (\gl_n)$, and let $C_L$ be the corresponding spectral curve. As is well known, if the spectral curve $C_L$ is non-singular, then so is $L$ (see Corollary \ref{smoothCor} below). However, if the spectral curve has singularities, then it is, in general, not possible to say whether $L$ is singular or not. The crucial point is that the spectral curve is the same for all points $L$ lying in the same fiber. At the same time, integrable systems may (and, as a rule, do) have fibers containing both singular and non-singular points. 
 So, to distinguish between different types of points in the same fiber, we introduce the notion of an \textit{essential} singular point of the spectral curve.\par

 Let a matrix polynomial $L \in \pazocal P_d^J (\gl_n)$ be such that the corresponding spectral curve $C_L$ is nodal. Denote the set of all double points of the curve $C_L$ by $\Sing C_L$. 
  By definition of the spectral curve, for any point $ (\lambda, \mu) \in C_L$ one has $\dim \Ker(L(\lambda) - \mu  \E) > 0$. Moreover, it is well known that if the point $ (\lambda, \mu) \in C_L$ is non-singular, then $\dim \Ker(L(\lambda) - \mu  \E) = 1$, i.e. the $\mu$-eigenspace of the matrix $L(\lambda)$ is one-dimensional (see, e.g., Chapter~5.2 of~\cite{babelon}). Similarly, if $ (\lambda, \mu) \in \Sing C_L$ is a double point, then the space $ \Ker(L(\lambda) - \mu  \E)$ can be of dimension either $1$ or $2$.
  \begin{definition}\label{essentialPoint}
A double point $(\lambda, \mu) \in \Sing C_L$ is called \textit{$L$-essential} if $\dim \Ker(L(\lambda) - \mu  \E) = 2$.
  \end{definition}
We denote the set of $L$-essential double points by $\pazocal E_L$. Note that this set depends on $L$, and not only on the spectral curve $C_L$. Also note that essential double points are preserved by each of the flows~\eqref{loopI}. \par
Now, we are in a position to formulate the first main result of the paper.
\begin{theorem}\label{thm1}
Assume that $L \in \pazocal P_d^J (\gl_n)$ is such that the corresponding spectral curve $C_L$ is nodal. Then: \begin{enumerate}
\item  The corank of $L$ is equal to the number $|\pazocal E_L|$ of $L$-essential double points of the spectral curve. In particular, $L$ is non-singular if and only if all double points of the spectral curve are not $L$-essential.
\item If the point $L$ is singular, then it is non-degenerate.
\end{enumerate}
\end{theorem}
In particular, one has the following well-known result.
\begin{corollary}\label{smoothCor}
If the spectral curve is smooth, then the corresponding fiber is non-singular.
\end{corollary}
The proof of Theorem~\ref{thm1} is given in Section~\ref{sec:pt1}. The main ingredients of the proof are Lemmas~\ref{lemma1} and~\ref{lemma2} proved in Sections~\ref{sec:lf} and~\ref{sec:evvr} respectively.
\begin{remark}\label{manyPoisson}
Recall that the notions of non-degeneracy and corank are defined for integrable systems on \textit{symplectic} manifolds (see Section~\ref{sec:isnds}). So, when saying that $L$ is a non-degenerate point for the polynomial matrix system, what we mean is that $L$ is non-degenerate for the \textit{restriction} of this system to the symplectic leaf of the Poisson bracket~\eqref{PoissonI}.  
Note that the proof of Theorem~\ref{thm1} does not use the explicit form of the Poisson bracket, but only the fact that the symplectic leaf containing $L$ is of maximal dimension (which is always the case for the bracket given by~\eqref{PoissonI}, see Proposition~\ref{leaves} below). Therefore, the conclusion of the theorem is, in fact, true for any Poisson bracket associated with the polynomial matrix system (see Remark \ref{multiHam}), provided that one can check the regularity of this bracket at $L$ (cf. Theorem~3 of~\cite{bolsinov2014singularities}).
\end{remark}

\begin{example}\label{gl2ex}
 Consider the polynomial matrix system on the space $\pazocal P_1^J(\gl_2)$, where $J$ is diagonal $J = \mathrm{diag}(j_1, j_2)$. Let also $L_0$ be a diagonal matrix $L_0 =  \mathrm{diag}(a_1, a_2)$. Then $L = L_0 + \lambda J \in \pazocal P_1^J(\gl_2)$ is a rank zero singular point for the polynomial matrix system (more generally, one can prove that $L$ is a rank zero singular point of a polynomial matrix system if and only if $L$ commutes with its leading coefficient $J$). The corresponding spectral curve is the union of two straight lines $\mu = a_1 + j_1 \lambda$ and $\mu = a_2 + j_2 \lambda$, so, by Theorem \ref{thm1}, the singular point $L$ is non-degenerate. Therefore, by Theorem~\ref{EliassonThmComplex}, the fiber of the polynomial matrix system containing $L$ is locally a union of two $1$-dimensional disks $D_1, D_2$ which intersect transversally at $L$. In this example, it is easy to describe these disks explicitly. They are given by shifted nilpotent subalgebras
\begin{align*}
D_1 =  \left\{\left. \left(\begin{array}{cc}a_1 + j_1\lambda  &z  \\0 & a_2 + j_2 \lambda \end{array}\right) \,\right| z\, \in \Complex  \right\}, \quad  D_2 = \left\{ \left.\left(\begin{array}{cc}a_1 + j_1 \lambda & 0 \\z & a_2 + j_2 \lambda\end{array}\right) \,\right| z\, \in \Complex   \right\}.
\end{align*}
However, in most cases, the structure of singular fibers is more complicated. Consider, for instance, a $3 \times 3$ generalization of the previous example: $J = \mathrm{diag}(j_1, j_2, j_3)$, $L_0 =  \mathrm{diag}(a_1, a_2, a_3)$, $L = L_0 + \lambda J \in \pazocal P_1^J(\gl_3)$. The corresponding spectral curve is the union of three lines  $\mu = a_i + j_i \lambda$. Assume that these lines are in general position. Then, by Theorem \ref{thm1}, the matrix polynomial $L$ is a rank zero non-degenerate singular point, and, by Theorem \ref{EliassonThmComplex}, the fiber of $L$ is locally a union of eight $3$-dimensional disks. Six of these disks are, as in the $2 \times 2$ example, shifted maximal nilpotents subalgebras (the number $6$ is the order the Weil group). At the same time, one can show that the two remaining disks lie on a single irreducible algebraic surface of high degree. This surface has a double point at $L$, giving rise to two local disks. This can also be generalized to the $n \times n$ situation. In this case, one can describe the structure of rank zero singular fibers in terms of multiplicities of weights for certain representations of $\gl_n$, see \cite{izosimov2015matrix} for details.
\end{example}

\begin{remark}\label{anyEss}
 From Theorem~\ref{thm1} it follows that {if the spectral curve is nodal, then \textit{all} singularities on the corresponding fiber are non-degenerate}. It can also be shown that for each nodal spectral curve $C$ and each subset of its nodes $\pazocal E \subset \Sing C$, there exists  $L \in \pazocal P_d^J (\gl_n)$ such that the corresponding set $\pazocal E_L$ is exactly $\pazocal E$. In particular, the corank of a non-degenerate singular fiber (i.e. the maximal corank of singular points on the fiber) is equal to the number of nodes in the spectral curve. Another corollary is that for complex polynomial matrix systems the set of singular values (the so-called \textit{bifurcation diagram}) coincides with the discriminant of the spectral curve. For degree one matrix polynomials, this is proved in \cite{Brailov, Konyaev}.
\end{remark}


%
%
%
\begin{remark}
We remark that Theorem~\ref{thm1} is not a criterion: a point $L$ can be non-degenerate for the polynomial matrix system even if the spectral curve has degenerate singularities. This is possible because degenerate singular points of the spectral curve may be not $L$-essential (in a certain generalized sense), or, roughly speaking, ``insufficiently'' $L$-essential. However, we conjecture that if the spectral curve has a non-nodal singular point, then there is at least one degenerate point on the corresponding fiber of the polynomial matrix system. 
\end{remark}
\begin{remark}
Note that the result of Theorem~\ref{thm1} is also true for \textit{Beauville systems} \cite{beauville1990jacobiennes}, which can be regarded as a particular case of \textit{Hitchin systems} (see the review~\cite{donagi1996spectral}). The Beauville system is defined as the quotient of the polynomial matrix system with respect to the conjugation action of the centralizer of $J$ in $\PGL_n(\Complex)$. Theorem~\ref{thm1} remains true for Beauville systems because non-degeneracy and corank of singularities are preserved under Hamiltonian reduction. 
\end{remark}

\medskip
\subsection{Williamson types of singular points}\label{sec:type}
In this section, we discuss non-degenerate singularities for real forms of polynomial matrix systems (see Section \ref{sec:rpms} for the definition of these real forms). Note that Theorem \ref{thm1} for complex polynomial matrix systems also holds for their real forms (see Proposition~\ref{complexification}), so it suffices to describe Williamson types of singular points.\par
Assume that $L$ is a matrix polynomial in $\pazocal P_d^J (\g)$, where $\g$ is a real form of $\gl_n$. Then, by Proposition~\ref{propInv}, the spectral curve $C_L$ is endowed with an antiholomorphic involution $\tau$  given by formula~\eqref{involution}. Thus, $C_L$ is a \textit{real algebraic curve} (recall that a real algebraic variety can be defined as a complex variety endowed with an antiholomorphic involution). 
\begin{definition}
The \textit{real part} of the spectral curve $C_L$ is the set of points of $C_L$ fixed by $\tau$. 
\end{definition}
\begin{remark}
 In the non-compact case, the real part is the set of points with real coordinates, while in the compact case the real part consists of points $(\lambda, \mu)$ with $\lambda \in \R$ and $\mu \in \sqrt{-1}\R$.
 \end{remark}
Now, assume that $C_L$ is a nodal curve, and let $P \in \Sing C_L$ be a node lying in its real part. Consider the two tangents to $C_L$ at $P$. Then there are two possible situations: either both tangents are fixed by $\tau$, or they are interchanged. In the first case, the real part of the curve locally looks like two smooth curves intersecting at $P$, while in the second case, the real part is locally a point (the point $P$).

\begin{definition}
Assume that a double point $P$ lies in the real part of $C_L$. Then, if the involution $\tau$ fixes the tangents to $C_L$ at $P$, the point $P$ is called \textit{a self-intersection point}. Otherwise, $P$ is called \textit{an isolated point}\footnote{There also exist somewhat archaic terms a \textit{crunode} for a self-intersection point and an \textit{acnode} for an isolated point. }.
\end{definition}
\begin{example}
The curve $\lambda^2 - \mu^2 = 0$, regarded as a real curve with a standard antiholomorphic involution $(\lambda, \mu) \mapsto (\bar \lambda, \bar \mu)$, has a self-intersection point at the origin, while the curve $\lambda^2 + \mu^2 = 0$ has an isolated double point. Note that these curves are also real with respect to the involution $(\lambda, \mu) \mapsto (\bar \lambda, -\bar \mu)$, and in terms of the latter involution the curve $\lambda^2 - \mu^2 = 0$ has an isolated point, while the curve $\lambda^2 + \mu^2 = 0$ has a self-intersection. \par
Such a situation with two or more involutions occurs in systems associated with loop algebras in types other than $A_n$, as well as in systems associated with twisted loop algebras (for example, the spectral curve of the Kovalevskaya top carries four antiholomorphic involutions). These systems will be discussed in detail in our forthcoming publication.
\end{example}
Apart from self-intersections and isolated points, a nodal spectral curve $C_L$ may have double points which do not lie in the real part of the curve. Such points necessarily come in pairs $P \leftrightarrow \tau(P)$. Note that for any matrix polynomial $L$ the points $P$ and $\tau(P)$ are either both $L$-essential, or both not $L$-essential.
 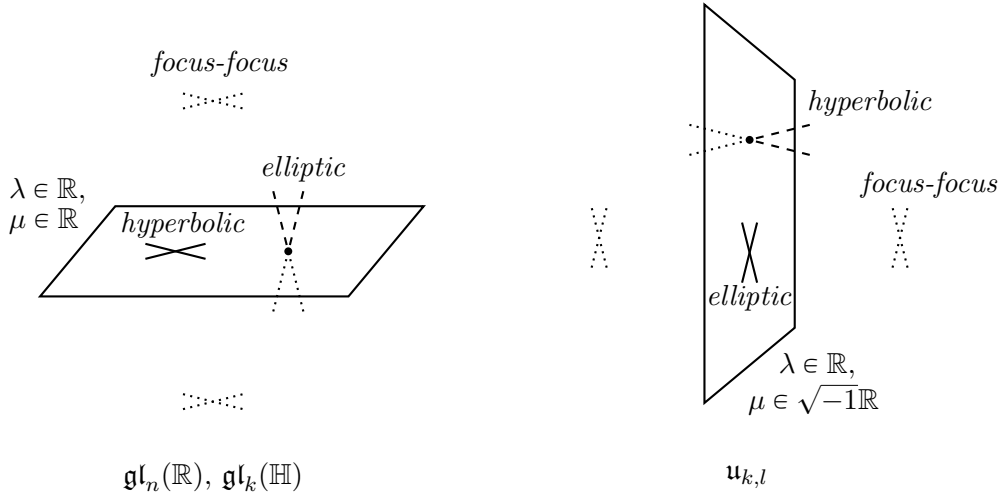
\begin{figure}[t]
 \centering
\begin{tikzpicture}[thick]
               \draw[dotted](3, -0.3) -- (3.2,0.5);
          \draw[dotted](3.4, -0.3) -- (3.2,0.5);
\draw (-0.1,-0.1) -- (0.9,1.1) -- (5,1.1) -- (4,-0.1) -- cycle;
    \fill (3.2,0.5) circle [radius=1.5pt];
  \draw(1.3, 0.4) -- (2.1,0.6);
   \draw(1.3, 0.6) -- (2.1,0.4);
     \draw[dotted](1.8, 2.4) -- (2.6,2.6);
   \draw[dotted](1.8, -1.4) -- (2.6,-1.6);
        \draw[dotted](1.8, -1.6) -- (2.6,-1.4);
   \draw[dotted](1.8, 2.6) -- (2.6,2.4);
     \draw[dashed](3, 1.3) -- (3.2,0.5);
          \draw[dashed](3.4, 1.3) -- (3.2,0.5);
          \node () at (1.8,0.85) {\textit{hyperbolic}};
            \node () at (3.4,1.6) {\textit{elliptic}};
                      \node () at (2.3,3) {\textit{focus-focus}};
            \node () at (2.2,-2.5) {$\gl_n(\R)$, $\gl_k(\H)$};  
                       \node () at (9.3,-2.5) {$\un_{k,l}$}; 
                       \node () at (0., 1.3) {$\lambda \in \R,$};
                            \node () at (-0.05, 0.9) {$\mu \in \R$};
            \node () at (10,1)
            {
            \begin{tikzpicture}[rotate = 90, yscale = -1, thick]
               \draw[dotted](3.3, -0.3) -- (3.5,0.5);
          \draw[dotted](3.7, -0.3) -- (3.5,0.5);
\draw(0,-0.1) -- (1,1.1) -- (4.3,1.1) -- (5.3,-0.1) -- cycle;
    \fill (3.5,0.5) circle [radius=1.5pt];
    \node () at (2,0.5) {
    \begin{tikzpicture}[rotate = 90, thick]
  \draw(1.1, 0.4) -- (1.9,0.6);
   \draw(1.1, 0.6) -- (1.9,0.4);
   \end{tikzpicture}
   };
   \node () at (1.4,0.5) {\textit{elliptic}}; 
     \draw[dashed](3.3, 1.3) -- (3.5,0.5);
          \draw[dashed](3.7, 1.3) -- (3.5,0.5);
            \node () at (4,2.1) {\textit{hyperbolic}}; 
               \node () at (2.9,2.9) {\textit{focus-focus}};         
       \draw[dotted](1.8, 2.4) -- (2.6,2.6);
   \draw[dotted](1.8, -1.4) -- (2.6,-1.6);
        \draw[dotted](1.8, -1.6) -- (2.6,-1.4);
   \draw[dotted](1.8, 2.6) -- (2.6,2.4);
      \node () at (0.5, 1.4) {$\lambda \in \R,$};
            \node () at (0.05, 1.35) {$\mu \in \sqrt{-1}\R$};
       \end{tikzpicture}
            };
\end{tikzpicture}
\caption{Singularities of spectral curves and singularities of integrable systems.}\label{singReal}
\end{figure}

\begin{theorem}\label{thm2}
Assume that $L \in \pazocal P_d^J (\g)$, where $\g$ is a real form of $\gl_n$, is a singular point of the real polynomial matrix system. Assume also that the corresponding spectral curve $C_L$ is nodal. 
Then the Williamson type of $L$ is $(k_e,k_h,k_f)$, where
\begin{enumerate}
\item  for real forms of non-compact type $k_e$ is the number of $L$-essential isolated points, and $k_h$ is the number of $L$-essential self-intersection points;
\item for real forms of compact type case $k_e$ is the number of $L$-essential self-intersection points, and $k_h$ is the number of $L$-essential isolated points;
\item  for any real form, $k_f$ is the number of pairs of $L$-essential non-real double points.
\end{enumerate}
See Figure~\ref{singReal}. 
\end{theorem}
Also note that for a given real form $\g$, some configurations of double points are impossible. In particular, one has the following corollary.
\begin{corollary}\label{unitary}
For polynomial matrix systems on $ \pazocal P_d^J (\un_n)$ and $ \pazocal P_d^J (\gl_k(\H))$, there are no singularities of hyperbolic type (i.e., all singularities are of type $(k_e,0,k_f)$).
\end{corollary}
\begin{proof}[Proof of the corollary]
First, let $L \in \pazocal P_d^J (\un_n)$. Then, since all eigenvalues of a skew-Hermitian matrix are pure imaginary, the equation $\det(L(\lambda) - \mu \E) = 0$, where $\lambda \in \R$ is fixed, has $n$ pure imaginary solutions, counting with multiplicities. Therefore, the real part of the spectral curve (given by $\lambda \in \R$, $\mu \in \sqrt{-1}\R$) does not have isolated points, proving that there are no hyperbolic singularities.\par
Now, consider the $\gl_k(\H)$ case. From the definition of this Lie algebra it easily follows that for each real eigenvalue of a matrix in $\gl_k(\H)$, there are at least two independent eigenvectors. Therefore, for $L \in \pazocal P_d^J (\gl_k(\H))$, there are no smooth points in the real part of the spectral curve  (recall that at smooth points we have $\dim \Ker (L(\lambda) - \mu \E) = 1$), and hence no self-intersection points. The result follows.
\end{proof}
\begin{remark}
Formally speaking, the above argument shows that polynomial matrix systems on $ \pazocal P_d^J (\un_n)$ and $ \pazocal P_d^J (\gl_k(\H))$ do not have hyperbolic singularities on \textit{fibers corresponding to nodal spectral curves}. However, since such fibers are dense in the set of all singular fibers, and hyperbolic singularities are stable under small perturbations, it follows that there can be no hyperbolic singularities at all.
\end{remark}
\begin{remark}
Also note that in the $\un_n$ and $\gl_k(\H)$ cases all double points in the real part of the spectral curve are automatically essential (cf. Remark \ref{anyEss}).
\end{remark}
For the proof of Theorem~\ref{thm2}, see Section~\ref{sec:pt2}. Some examples of the application of this theorem can be found in Section \ref{sec:examples}. \par
\bigskip
\section{Nodal curves and generalized Jacobians}\label{sec:nc}
In this section we give a brief geometric introduction to the theory of generalized Jacobians of nodal curves. Our point of view is close to the original approach of M.\,Rosenlicht~\cite{Rosenlicht1, Rosenlicht}, see also J.P.\,Serre~\cite{Serre}. In what follows, we will need these results on generalized Jacobians to prove Theorem~\ref{thm1}.
 
\subsection{Nodal curves and regular differentials}\label{sec:ncrd}
We begin with a definition of an ``abstract'' nodal curve (cf. Definition \ref{def:nc1} of a plane nodal curve). Let $\Gamma = \Gamma_1 \sqcup \ldots  \sqcup \Gamma_m$ be a disjoint union of connected Riemann surfaces, 
and let $ \Lambda = \{  \{ P_1^+, P_1^-\}, \dots,  \{P_k^+, P_k^-\} \} $ be a finite collection of pairwise disjoint $2$-element subsets of $\Gamma$. Consider the topological space $\Gamma\,/\, \Lambda$ obtained from $\Gamma$ by identifying $ P_i^+$ with $P_i^-$  for each $i = 1\, \dots, k$.  
Let $\pi \colon \Gamma \to \Gamma\,/\, \Lambda$ be the natural projection, and let $\mathrm{supp}( \Lambda) := \{P_1^+, P_1^-, \dots, P_k^+, P_k^-\}.$
\begin{definition} \label{def:nodal}
A function $f \colon \Gamma\,/\, \Lambda \to \CP^1$ is called \textit{meromorphic} on $\Gamma\,/\, \Lambda $ if its pullback $\pi^*f$ is a meromorphic function on $\Gamma$ which does not have poles at points in $\mathrm{supp}(\Lambda)$.
The space $\Gamma\,/\, \Lambda$ endowed with this ring of meromorphic functions is called a \textit{nodal curve}. \textit{Irreducible components} of the nodal curve  $\Gamma\,/\, \Lambda$ are, by definition, the subsets $\pi(\Gamma_i) \subset \Gamma\,/\, \Lambda$.
\end{definition}
\begin{remark}
Note that, {formally speaking, to turn $\Gamma\,/\, \Lambda$ into a complex analytic space, we should describe its structure sheaf; however, for the purposes of the present paper, it suffices to define global meromorphic functions}.
\end{remark}
 Obviously, a smooth compactification of any plane nodal curve in the sense of Definition~\ref{def:nc1} is a nodal curve in the sense of Definition~\ref{def:nodal} (by a smooth compactification of a plane nodal curve $C$ we mean a compact complex-analytic space $X$ such that $C$ is biholomorphic to $X$ minus a finite number of smooth points). 
\par
Note that any nodal curve $\Gamma\,/\,\Lambda$ may be regarded simply as a pair $(\Gamma, \Lambda)$ (with an appropriately defined ring of meromorphic functions). The latter point of view is very convenient when one needs to consider different partial normalizations of the same curve, and we adopt this point of view in the present paper. 
 \begin{definition}\label{srf}
 We say that a function $f$ {meromorphic on $\Gamma$} is $\Lambda$\textit{-regular} if for each $\{P_i^\pm\} \in \Lambda$ we have $f(P_i^-) = f(P_i^+) \neq \infty$.
  \end{definition}
  We denote the ring of $\Lambda$-regular functions on $\Gamma$ by $\pazocal M(\Gamma, \Lambda)$ (note that, in contrast to the case of smooth curves, the ring $\pazocal M(\Gamma, \Lambda)$ is not a field; moreover, if $\Gamma$ has several connected components, then the corresponding ring $\pazocal M(\Gamma, \Lambda)$ has zero divisors). Obviously, $\Lambda$-regular functions on $\Gamma$ are in one-to-one correspondence with meromorphic (in the sense of Definition \ref{def:nodal}) functions on $\Gamma\,/\, \Lambda$. 
 \begin{definition}\label{srd}
A meromorphic differential ($1$-form) $\omega$ on $\Gamma$ is $\Lambda$\textit{-regular} if all its poles are simple, contained in $\mathrm{supp}(\Lambda)$, and for all $ \{P_i^\pm\} \in \Lambda$ we have
\begin{align*}
 {\Res}_{P_i^+} \,\omega +  {\Res}_{P_i^-} \,\omega = 0 \,.\end{align*}
  \end{definition}
   We denote the space of $\Lambda$-regular differentials on $\Gamma$ by $\Omega^1(\Gamma,\Lambda)$. 
   \begin{remark}
Definition \ref{srd} can be reformulated as follows: a differential $\omega$ is $\Lambda$-regular if for any $\Lambda$-regular function $f \in \pazocal M(\Gamma, \Lambda)$, we have
$
\mathrm{Tr}_f\, \omega \equiv 0
$. 
Vice versa, a meromorphic function $f$ on $\Gamma$ which does not have poles at points of $\mathrm{supp}(\Lambda)$ is $\Lambda$-regular if and only if $\mathrm{Tr}_f\, \omega \equiv 0$ for any $\omega \in \Omega^1(\Gamma,\Lambda)$.
\end{remark}
   The dimension of the space $\Omega^1(\Gamma,\Lambda)$ can be computed in terms of the \textit{dual graph} of $\Gamma\,/\, \Lambda$. 
   \begin{definition}
The dual graph of a nodal curve $\Gamma\,/\, \Lambda$ is the graph whose vertices are components $\Gamma_1, \dots, \Gamma_m$ of $\Gamma$, and whose edges are in one-to-one correspondence with elements of $\Lambda$. Namely, each pair  $ \{P_i^\pm\} \in \Lambda$ with $P_i^- \in \Gamma_s$ and $P_i^+ \in \Gamma_t$ gives rise to an edge joining $\Gamma_s$ to $\Gamma_t$ (see Figure~\ref{curveGraph}).
\end{definition}
 \par 
    \begin{figure}[t]
 \centering
\begin{tikzpicture}[thick]
\draw (0,0) -- (1.5,2);
\draw (2,0) -- (0.5,2);
\draw (-0.2,0.5) -- (2.2,0.5);
\draw[->, dashed] (1.8,1.2) -- (2.7,1.2);
\node () at (3.5,1) {
\begin{tikzpicture}[rotate = 180, thick]
 \draw   (0,0) -- (0.5, 0.86);
    \fill (0,0) circle [radius=1.5pt];
       \fill (1,0) circle [radius=1.5pt];
          \fill (0.5,0.86) circle [radius=1.5pt];
   		 \draw  (0,0) -- (1,0);    
 		   \draw    (0.5, 0.86) -- (1,0); 
		   \end{tikzpicture}
};
\node () at (7,1.2) {\includegraphics[scale = 1]{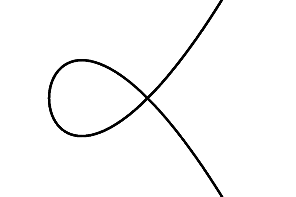}};
\draw[->, dashed] (7.8,1.2) -- (8.7,1.2);
\node () at (9.5,1) {
\begin{tikzpicture}[thick]
    \fill (0,0) circle [radius=1.5pt];
    \draw (0,0) .. controls +(-0.7,-0.7) and +(+0.7,-0.7) .. (0,0);
		   \end{tikzpicture}
};
\end{tikzpicture}
\caption{Nodal curves and their dual graphs.}\label{curveGraph}
\end{figure}
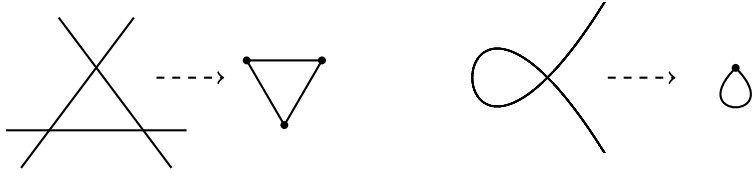
\begin{proposition}\label{regDiff} 
The dimension of the space
  $ \Omega^1(\Gamma,\Lambda)$ is equal to   the genus of $\Gamma$ (i.e, the sum of genera of $\Gamma_i$'s) plus the first Betti number of  the dual graph of $\Gamma\, / \,\Lambda$.
\end{proposition}
When the curve $\Gamma\,/\, \Lambda$ is connected (which is equivalent to saying that the dual graph is connected), the number $\dim \Omega^1(\Gamma,\Lambda)$ is called the \textit{arithmetic genus} of $\Gamma\,/\, \Lambda$. 
In the case when a nodal curve $\Gamma\, / \,\Lambda$ is a smooth compactification of a plane nodal curve $C$, the arithmetic genus of $\Gamma\, / \,\Lambda$ can be found by counting integer points in the interior of the Newton polygon of $C$.
\begin{definition}
Let $C \in \Complex^2$ be a plane curve given by the polynomial equation $ \sum   a_{ij}\lambda^i \mu^j = 0$. 
Then the \textit{Newton polygon} $\Delta$ of $C$ is the convex hull of the set
$
\{ (i,j) \in \Z^2 \mid a_{ij} \neq 0\} .
$
\end{definition}
Let $C \subset \Complex^2$ be a plane curve. Consider any side $\delta \in \Delta$ of the corresponding Newton polygon, and let $(i_0, j_0), \dots, (i_p,j_p)$ be the consecutive integer points lying on $\delta$. Let
\begin{align}\label{fDelta}
f_\delta(z) := \sum\nolimits_{i=0}^p a_{i_p,j_p}z^p\,.
\end{align}
\begin{definition}\label{ndSide}
The curve $C$ is called \textit{non-degenerate with respect to a side $\delta$ of its Newton polygon} if the corresponding polynomial $f_\delta(z) $ has no multiple roots. The curve $C$ is called \textit{non-degenerate with respect to its Newton polygon $\Delta$} if it is non-degenerate with respect to all sides of $\Delta$.

\end{definition}
The following result is due to A.G.\,Khovanskii.
\begin{theorem}\label{KhThm}
Let $C$ be a plane nodal curve not passing through the origin and non-degenerate with respect to its Newton polygon. Then the arithmetic genus of the smooth compactification of $C$ is equal to the number of  integer points in the interior of the Newton polygon of $C$.
\end{theorem}
For smooth curves, this is proved in Section 4 of \cite{khovanskii1978newton}. The general case can be proved using the fact that the arithmetic genus is constant in families.

\begin{remark}
One can also show that if a curve $f(\lambda, \mu) = 0$ satisfies the assumptions of Theorem~\ref{KhThm}, then regular differentials on  its smooth compactification are given by
$$
\omega_{ij} = \frac{\lambda^{i-1}\mu^{j-1}\diff \lambda}{\partial_\mu f },
$$
where $(i,j) \in \Z^2$ belongs to the interior of the Newton polygon.
\end{remark}
In fact, the non-degeneracy condition in Theorem \ref{KhThm} can be weakened, provided that the Newton polygon has the following additional property.
\begin{definition}
We say that a Newton polygon is \textit{northeast facing} if for any of its sides not lying on coordinate axes the outward normal vector has non-negative coordinates (see Figure \ref{nef}).
\end{definition}

    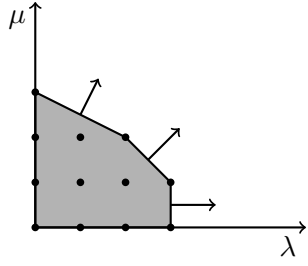
\begin{figure}[t]
\centerline{
\begin{tikzpicture}[scale = 1.2, thick]
   		 \draw  [->] (0,0) --(0,2.5);
		  \draw  [->] (0,0) --(3,0);
\draw (0,0) -- (1.5,0)  -- (1.5,0.5) -- (1,1) -- (0,1.5) -- cycle;
		  \fill[opacity = 0.3] (0,0) -- (1.5,0)  -- (1.5,0.5) --  (1,1) -- (0,1.5) -- cycle;
		  \draw [->] (1.5,0.25) -- (2, 0.25);
		  	  \draw [->] (1.25,0.75) -- (1.6, 1.1);
			    \draw [->] (0.5, 1.25) -- (0.7, 1.65);
		  \fill (0,0) circle (1.2pt);
		    \fill (0,0.5) circle (1.2pt);
		      \fill (0,1) circle (1.2pt);
		        \fill (0,1.5) circle (1.2pt);
		        		  \fill (0.5,0) circle (1.2pt);
		    \fill (0.5,0.5) circle (1.2pt);
		      \fill (0.5,1) circle (1.2pt);	
		      		        		  \fill (1,0) circle (1.2pt);
		    \fill (1,0.5) circle (1.2pt);
		      \fill (1,1) circle (1.2pt);
		      		        		  \fill (1.5,0) circle (1.2pt);
		    \fill (1.5,0.5) circle (1.2pt);
														    		    \node () at (-0.2, 2.3) {$\mu$};
																    \node () at (2.8, -0.2) {$\lambda$};   
\end{tikzpicture}
}
\caption{A northeast facing Newton polygon.}\label{nef}
\end{figure}
\begin{theorem}\label{KhThm2}
Let $C$ be a plane nodal curve whose Newton polygon $\Delta$ is northeast facing. Assume that $C$ is non-degenerate with respect to all sides of $\Delta$ not lying on coordinate axes. Then the arithmetic genus of the smooth compactification of $C$ is equal to the number of  integer points in the interior of $\Delta$.
\end{theorem}
Theorem \ref{KhThm2} is reduced to Theorem \ref{KhThm} by a linear change of coordinates $(\lambda, \mu) \mapsto (\lambda + \eps, \mu + \eps )$. The northeast facing condition guarantees that the Newton polygon is invariant under this transformation. At the same time, for generic $\eps$, the shifted curve will satisfy the non-degeneracy condition for all sides of the Newton polygon, including those which lie on the coordinate axes.
\medskip
\subsection{Generalized Jacobians and Picard groups}\label{sec:AJT}
In this section, we recall the notions of the generalized Jacobian and the generalized Picard group of a nodal curve.
Consider the mapping $ \mathcal I \colon \Hom_1(\Gamma\,\setminus\, \mathrm{supp}(\Lambda), \Z) \to \Omega^1(\Gamma,\Lambda)^*$ given by
\begin{align*}
\langle \mathcal I(\gamma), \omega \rangle := \oint_\gamma\, \omega\quad \forall \, \omega \in \Omega^1(\Gamma,\Lambda)\,.
\end{align*} The image of $\mathcal I$ is a lattice $L(\Gamma, \Lambda) \subset \Omega^1(\Gamma,\Lambda)^*$, called the \textit{period lattice}.
\begin{definition}
The quotient $\Jac(\Gamma, \Lambda) := {\Omega^1(\Gamma,\Lambda)^*}\,/\,{L(\Gamma, \Lambda)}$ is called the \textit{generalized Jacobian} of the nodal curve $\Gamma\,/\, \Lambda$. 
\end{definition}
By definition, the generalized Jacobian is an Abelian complex Lie group of dimension $\dim \Omega^1(\Gamma,\Lambda)$. The following proposition gives a more precise description.
\begin{proposition}
These exists an exact sequence of complex Lie group homomorphisms
$$
0 \to (\Complex^*)^b \to \Jac(\Gamma,\Lambda) \to \Jac(\Gamma) \to 0\,,
$$
where $b $ is the first Betti number of the dual graph, and $\Jac(\Gamma)$ is the usual Jacobian of $\Gamma$ (i.e., the product of Jacobians of $\Gamma_i$'s). Geometrically, $\Jac(\Gamma,\Lambda)$ is a principal $(\Complex^*)^{b}$-bundle over $\Jac(\Gamma)$.
\end{proposition}
Now, we define the notion of a regular divisor on a nodal curve and construct the generalized Picard group. As above, let $m$ be the number of connected components of $\Gamma$ (i.e., the number of irreducible components of the curve $\Gamma\,/\,\Lambda$). For each Weil divisor $D$ on $\Gamma$, define its \textit{multidegree} $\mdeg D = (\gamma_1, \dots, \gamma_{m}) \in \Z^{m}$ by setting $\gamma_i := \deg \left(D\vert_{\Gamma_i}\right)$. The total degree of $D$ is the number $\deg D := \sum\nolimits_{i=1}^{m} \gamma_i.$  Denote the set of divisors of multidegree $\gamma$ on $\Gamma$ by $\Div_\gamma(\Gamma)$. 
\begin{definition}A divisor $D$ on $\Gamma$ is called \textit{$\Lambda$-regular} if its support does not intersect $\mathrm{supp}(\Lambda)$. 
\end{definition}
The set of $\Lambda$-regular divisors of multidegree $\gamma$ is denoted by $\Div_\gamma(\Gamma, \Lambda)$, while $\Div(\Gamma, \Lambda)$ stands for the set of all $\Lambda$-regular divisors. The latter is a $\Z^m$-graded Abelian group.\par 
Further, let $\pazocal M^*(\Gamma, \Lambda)$ be the set of invertible elements in the ring $\pazocal M(\Gamma, \Lambda)$. This set consists of functions which do not vanish at points of $\mathrm{supp}(\Lambda)$ and whose restriction to any of the components $\Gamma_i$ of $\Gamma$ does not vanish identically.
Clearly, for each $f \in \pazocal M^*(\Gamma, \Lambda)$, the corresponding divisor $$(f):= (\mbox{zeroes of } f) - (\mbox{poles of } f)$$ is $\Lambda$-regular. 
\begin{definition}Divisors of the form $(f)$, where $f \in \pazocal M^*(\Gamma, \Lambda)$, will be called \textit{$\Lambda$-principal}. Two $\Lambda$-regular divisors are $\Lambda$-\textit{linearly equivalent} if their difference is a $\Lambda$-principal divisor. 
\end{definition}
Denote the set of $\Lambda$-principal divisors by $\PDiv(\Gamma, \Lambda)$. Let also $[D]_\Lambda$ be the $\Lambda$-linear equivalence class of a $\Lambda$-regular divisor $D$.
\begin{definition}
The \textit{generalized Picard group} is
 $\Pic(\Gamma,\Lambda) :=   \Div(\Gamma, \Lambda) \,/  \,\PDiv(\Gamma, \Lambda).  $
 \end{definition}
The latter is a $\Z^m$-graded Abelian group:
$$
\Pic(\Gamma,\Lambda) = \bigsqcup_{\gamma \in \Z^{m}} \Pic_\gamma(\Gamma,\Lambda),
$$
 where $\Pic_\gamma(\Gamma,\Lambda)$ is the set of $\Lambda$-regular divisors of multidegree $\gamma$ modulo linear equivalence. \par
 The Abel map for nodal curves is defined in the same way as for smooth ones. Namely, let $D$ be a $\Lambda$-regular divisor of multidegree $0$. Then $D$ can be written as
$$
D = \sum\nolimits_{i=1}^{m}(D_i^+ - D_i^-)\,,
$$
where $D_i^{\pm}$ are effective divisors on $\Gamma_i$, and $\deg D_i^{+} = \deg D_i^{-}$. For a $\Lambda$-regular differential \nolinebreak$\omega$, let
$$
\int_D \omega :=\, \sum\nolimits_{i=1}^{m}\int_{D_i^-}^{D_i^+}\!\!\omega\,.
$$
Since this integral is defined up to periods of $\omega$, we obtain a homomorphism \begin{align*}\Div_0(\Gamma,\Lambda) &\to \Jac(\Gamma, \Lambda)\,\\
D &\mapsto \int_D \omega\,,
\end{align*} which is the analog of the usual Abel map. In the same way as in the smooth case, one shows that $\Lambda$-principal divisors lie in the kernel of the Abel map, therefore the Abel map descends to a homomorphism $ \Pic_0(\Gamma,\Lambda) \to \Jac(\Gamma, \Lambda)$. Moreover, one has the following analog of the Abel-Jacobi theorem.
\begin{theorem}
The Abel map is an isomorphism between $\Pic_0(\Gamma,\Lambda) $ and $\Jac(\Gamma,\Lambda)$. 
\end{theorem}

 From this theorem it follows that for any multidegree $\gamma$ the set $\Pic_\gamma(\Gamma,\Lambda)$ is a principal homogeneous space (a torsor) of the group $\Pic_0(\Gamma,\Lambda)  \simeq \Jac(\Gamma,\Lambda)$. In particular, $\Pic_\gamma(\Gamma,\Lambda)$ has a canonical affine structure, and the tangent space to $\Pic_\gamma(\Gamma,\Lambda)$  at any point can be naturally identified with the dual  $\Omega^1(\Gamma,\Lambda)^*$ of the space of $\Lambda$-regular differentials.
 
 \medskip
 \subsection{On more general curves}\label{sec:gnc}
In what follows, we will need a slightly more general class of curves which we call \textit{generalized nodal curves}. As before, let $\Gamma = \Gamma_1 \sqcup \ldots \sqcup \Gamma_m$ be a disjoint union of connected Riemann surfaces, and let $\Lambda = \{ \mathbf P_1, \dots, \mathbf P_k\}$ be a collection of pairwise disjoint finite  (not necessarily $2$-element)  subsets of $\Gamma$. A generalized nodal curve $\Gamma\,/\, \Lambda$ is obtained from $\Gamma$ by gluing points within each $\mathbf P_i$ into a single point. The corresponding ring of meromorphic functions is defined as
$$
\pazocal M(\Gamma,\Lambda) := \{ f \mbox{ meromorphic on } \Gamma \mid  P,Q \in \mathbf P_i \Rightarrow f(P)=f(Q) \neq \infty\},
$$
and $\Lambda$-regular differentials are those which are holomorphic outside $\mathrm{supp}(\Lambda) :=  \mathbf P_1 \cup \ldots \cup \mathbf P_k$ and satisfy
$$
\sum\nolimits_{P \in \mathbf P_i}\Res_P\,\omega = 0 
$$
for each $i = 1,\dots, k$.
The study of such curves can be reduced to nodal ones in the following way. Assume that $\mathbf P_i = \{P_1, \dots, P_s\}$. Consider a Riemann sphere $\CP^1$ with $s$ marked points $Q_1, \dots, Q_s$. Let $ \Gamma' := \Gamma \sqcup \CP^1$, and let $$ \Lambda' := \left(\Lambda \,\setminus\, \{\mathbf P_i\}\right) \cup \{ \{P_1, Q_1\}, \dots, \{P_s,Q_s\}\}.$$
Then the curve $\Gamma'\,/\,\Lambda'$ is ``equivalent'' to $\Gamma\,/\,\Lambda$ in the sense that there are natural identifications
$$
\pazocal M(\Gamma,\Lambda) \simeq \{ f \in \pazocal M(\Gamma',\Lambda')\} \mid f\vert_{\CP^1} = \const \}, \quad \Omega^1(\Gamma,\Lambda) \simeq \Omega^1(\Gamma',\Lambda')\,.
$$
These identifications allow one to reduce the study of $\Gamma\,/\,\Lambda$ to $\Gamma'\,/\,\Lambda'$. Repeating this procedure for all $\mathbf P_i$'s gives a nodal curve, which shows that all results of Section \ref{sec:nc} are true for generalized nodal curves as well.
\begin{remark}
Generalized nodes (also known as \textit{seminormal singularities}) are not to be confused with \textit{ordinary multiple points} (recall that a plane curve singularity is called an ordinary multiple point if the curve locally looks like $n$ smooth curves intersecting each other transversally). In contrast to ordinary multiple points, generalized nodes are non-planar. 
\end{remark}

\medskip
\section{Main technical lemmas}
\subsection{Linearization of flows on the generalized Jacobian}\label{sec:lf}\label{sec:lf1}\label{sec:cons}

It is well-known that if the spectral curve $C$ is smooth, then flows~\eqref{loopI} restricted to the corresponding fiber of the polynomial matrix system linearize on the Jacobian of  $C$ (see, e.g., Chapter~5.4 of~\cite{babelon}). In this section we prove a similar result for nodal curves.\par
Let $C$ be a nodal plane curve, and let $\pazocal E \subset \Sing C$ be a subset of its nodes. Consider the set
\begin{align}\label{ticieps}
\pazocal T_{C, \pazocal E} := \{ L \in  \pazocal P_d^J (\gl_n) \mid C_L = C, \,\pazocal E_L = \pazocal E \}\,
\end{align}
of matrix polynomials with a fixed spectral curve and fixed essential double points. Then, since the flows given by~\eqref{loopI} preserve the spectral curve and $L$-essential double points, it follows that the set $\pazocal T_{C, \pazocal E} $ is an invariant subset for the polynomial matrix system.

\par
Now, starting with a curve $C$, we  construct another curve 
whose (generalized) Jacobian is suitable for linearizing flows~\eqref{loopI} on the invariant set $\pazocal T_{C, \pazocal E} $. Let  $\Gamma$ be the non-singular compact model of the spectral curve $C$. 
It is a (possibly disconnected) Riemann surface obtained from $C$ by resolving all double points and adding smooth points $\infty_1, \dots, \infty_n$ at infinity. Let $\Gamma_\infty := \{ \infty_1, \dots, \infty_n\}$, and let $\pi \colon \Gamma \, \setminus \, \Gamma_\infty \to C$ be the natural projection. 
Further, let $ \pazocal{\bar E} := \Sing C\, \setminus \,\pazocal E$ be the set of non-essential double points of the spectral curve, and let
\begin{align}\label{Lambda}
 \Lambda := \{ \pi^{-1}(P) \mid P \in \pazocal{\bar E}\}.
 \end{align}
Then $\Gamma\, / \, \Lambda$ is a nodal curve which is obtained from the spectral curve $C$ by compactifying the latter at infinity and normalizing it at all essential double points.
Further, let \begin{align}\label{LambdaPrime}\Lambda' := \Lambda \cup \{\Gamma_\infty \},\end{align} 
and consider the curve $\Gamma \, / \, \Lambda'$. It is a generalized nodal curve which is obtained from $\Gamma\, / \, \Lambda$ by, roughly speaking, gluing the points at infinity into one point. \par With these definitions one can show that the set $\pazocal T_{C, \pazocal E} $ is biholomorphic to an open dense subset in the union of finitely many connected components of $ \Pic(\Gamma, \Lambda')$. Here we prove a weaker, local, result, which is nevertheless sufficient for our purposes. 
Fix a matrix polynomial $L_0 \in \pazocal T_{C, \pazocal E} $. Denote by $ O(L_0)$ the germ of the orbit of $L_0$ under the local Hamiltonian $\Complex^k$-action defined by vector fields~\eqref{loopI} (these vector fields commute, and only finitely many of them are linearly independent, so we indeed have a well-defined local $\Complex^k$-action). Clearly, we have $O(L_0) \subset \pazocal T_{C, \pazocal E}$. Furthermore, $O(L_0)$ is, by definition, smooth and invariant under flows~\eqref{loopI}.  \par

%
Further, consider the space $\Omega^1_m(\Gamma)$ of meromorphic forms on the Riemann surface $\Gamma$. Define a bilinear pairing 
$\InnerProduct_\infty \colon \Complex[\lambda^{\pm 1}, \mu] \times \Omega^1_m(\Gamma) \to \Complex
$
 by the formula
\begin{align}\label{inftyForm}
\langle \phi, \omega \rangle_\infty := \sum\nolimits_{i=1}^n \Res_{\infty_i} \phi(\lambda, \mu)\omega\,.
\end{align}
\begin{lemma}\label{lemma1}
There exists a locally biholomorphic map (the so-called \textit{eigenvector map}) $\Psi \colon O(L_0)\to \Pic(\Gamma, \Lambda')$ such that for any $\phi \in \Complex[\lambda^{\pm 1}, \mu]$, the image of the vector field given by~\eqref{loopI} under $\Psi$ is a translation-invariant vector field given by
\begin{align}
\label{velocity}
\omega\left(\diffFXYp{\xi}{t}\right) = -\langle \phi, \omega \rangle_\infty\,,
\end{align}
where $\xi \in \Pic(\Gamma, \Lambda')$, and $\omega$ is any $\Lambda'$-regular differential on $\Gamma$.
\end{lemma}
\begin{remark}
The expression on the left-hand side of formula~\eqref{velocity} is well-defined because the cotangent space to $\Pic(\Gamma, \Lambda')$ is naturally identified with the space of $\Lambda'$-regular differentials on $\Gamma$ (see Section~\ref{sec:AJT}). Furthermore, thanks to this identification, formula~\eqref{velocity}  uniquely determines the velocity vector $ \diff\xi / \diff t$. 
\end{remark}
\begin{remark}
Formula~\eqref{velocity} first appeared in the work of Mumford and Van Moerbeke \cite{MVM}. The result of Lemma~\ref{lemma1} in the smooth case is due to Gavrilov~\cite{Gavrilov2}.
\end{remark}
    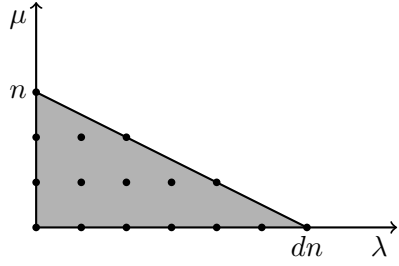
\begin{figure}[t]
\centerline{
\begin{tikzpicture}[scale = 1.2, thick]
   		 \draw  [->] (0,0) --(0,2.5);
		  \draw  [->] (0,0) --(4,0);
		  \draw (3,0) -- (0,1.5);
		  \fill[opacity = 0.3] (0,0) -- (3,0) -- (0,1.5) -- cycle;
		  \fill (0,0) circle (1.2pt);
		    \fill (0,0.5) circle (1.2pt);
		      \fill (0,1) circle (1.2pt);
		        \fill (0,1.5) circle (1.2pt);
		        		  \fill (0.5,0) circle (1.2pt);
		    \fill (0.5,0.5) circle (1.2pt);
		      \fill (0.5,1) circle (1.2pt);		       
		      		        		  \fill (1,0) circle (1.2pt);
		    \fill (1,0.5) circle (1.2pt);
		      \fill (1,1) circle (1.2pt);
		      		        		  \fill (1.5,0) circle (1.2pt);
		    \fill (1.5,0.5) circle (1.2pt);
		      	 		      		        		  \fill (2,0) circle (1.2pt);
		    \fill (2,0.5) circle (1.2pt);
		    		      		        		  \fill (2.5,0) circle (1.2pt);
		    \fill (3,0) circle (1.2pt);
		              \node () at (-0.2, 1.5) {$n$};
												    		    \node () at (3, -0.2) {$dn$};
														    		    \node () at (-0.2, 2.3) {$\mu$};
																    \node () at (3.8, -0.2) {$\lambda$};   
\end{tikzpicture}
}
\caption{Newton polygon of the spectral curve.}\label{newton}
\end{figure}

Before we prove Lemma \ref{lemma1}, let us show how it can be used to compute the rank of $L$. Notice that the rank of $L$ is, by definition, equal to the dimension of the orbit $ O(L)$, which, by Lemma \ref{lemma1}, is equal to the dimension of $\Pic(\Gamma, \Lambda')$.
\begin{proposition}\label{genusFormula}
We have
\begin{align}
 \dim \Pic(\Gamma, \Lambda') = \frac{1}{2}dn(n-1) - |\pazocal E_L|\,.
\end{align}
\end{proposition}
\begin{proof}
Without loss of generality, we may assume that the spectral curve $C$ does not pass through the origin (if it does, we apply a linear change of variables), and that the matrix $J$ does not have zero eigenvalues (if it does, we replace $J$ by $J + \eps \Id$; clearly, the resulting spectral curve is isomorphic to the initial one). Under these assumptions, the Newton polygon $\Delta$ of $C$ is the triangle depicted in Figure~\ref{newton}. Furthermore, from the simplicity of the spectrum of $J$ it follows that $C$ is non-degenerate with respect to the hypothenuse of $\Delta$ in the sense of Definition \ref{ndSide} (the corresponding polynomial $f_\delta$ is exactly the characteristic polynomial of $J$). Therefore, one can apply Theorem \ref{KhThm2}.\par Let  $\Gamma \, / \,\Lambda''$ be the smooth compactification of $C$. Then, by Theorem \ref{KhThm2}, the arithmetic genus of $\Gamma \, / \,\Lambda''$ is equal to the number of integer points in the interior of the triangle $\Delta$:
$$
 \dim \Pic(\Gamma, \Lambda'') = \frac{1}{2}dn(n-1) - n + 1\,.
$$

 Further, recall that the curve $\Gamma \, / \, \Lambda'$ is obtained from $\Gamma \, / \,\Lambda''$ by normalizing the latter at essential double points and identifying the points at infinity. The former operation corresponds to removing $|\pazocal E|$ edges from the dual graph, while the latter corresponds to adding one vertex and $n$ edges (see Section~\ref{sec:gnc}). So, from Proposition \ref{regDiff} it follows that
$$
 \dim \Pic(\Gamma, \Lambda') =  \dim \Pic(\Gamma, \Lambda'') - |\pazocal E|  - 1 + n = \frac{1}{2}dn(n-1) - |\pazocal E_L|\,,
$$
as desired.
\end{proof}
\begin{remark}
The genus of  $\Gamma \, / \, \Lambda'$ can also be found by first applying the Riemann-Hurwitz formula for the function $\lambda \colon \Gamma \to \bar \Complex$ to find the geometric genus of $\Gamma$, and then using Proposition \ref{regDiff}. 
\end{remark}



\begin{proof}[Proof of Lemma \ref{lemma1}] \textit{Step 1. Construction of the eigenvector map $\Psi$.} 
The construction of the map $\Psi$ is very similar to the case of a smooth spectral curve (see Chapter 5.2 of \cite{babelon}). 
Take any matrix polynomial $  L(\lambda) \in O(L_0)$. Then, for almost all points $P \in \Gamma$, we have $\dim \Ker(  L(\lambda(P)) - \mu(P)\E) = 1$. This gives a densely defined map 
\begin{align*}\psi_{L} \colon \Gamma &\to \CP^{n-1} \\
P &\mapsto \Ker\left( L(\lambda(P)) - \mu(P)\E\right) ,
\end{align*}
which, as any rational map of a curve to a projective space, extends to the whole of $\Gamma$. By continuity, we have 
$$
(L(\lambda) - \mu\E)\psi_{L} = 0
$$
at all points of $\Gamma$. 
 \par
Further, take a point ${\bf{a}} = (a_1, \dots, a_n) \in (\Complex^*)^n$, and let $h_{L} \colon \Gamma \to (\CP^1)^n$ be a meromorphic vector-function given by
\begin{align}\label{hFormula}
h_{L} := \frac{\psi_{L}}{\sum\nolimits_{i=1}^n a_i\psi_{L}^i}\,,
\end{align}
where $(\psi_{L}^1(P) : \dots : \psi_{L}^n(P))$ are homogenous coordinates of the point $\psi_{L}(P) \in \CP^{n-1}$. Denote by $D_{L}$ the pole divisor of $h_{L}$. 
Obviously,   for any $L \in O(L_0)$ one can choose the vector ${\bf a} \in (\Complex^*)^n$ in such a way that the divisor $D_{L}$ is simple (i.e., all its points are geometrically distinct) and $\Lambda'$-regular.  Moreover, if the orbit $O(L_0)$ is sufficiently small, one can choose ${\bf a} \in (\Complex^*)^n$ such that this property holds for all $L \in O(L_0)$ simultaneously. Indeed, $O(L_0)$ is defined as the joint orbit of flows of the form~\eqref{loopI}.
 In a standard way (see, e.g., Chapter 5.4 of~\cite{babelon}), one shows that when $L$ evolves according to~\eqref{loopI}, the evolution of its eigenvector $h_{L}$ is given by
\begin{align}
\label{heq}
\diffXp{t} h_{L} +  \phi(\lambda, L)_+ h_{L}= \eta h_{L}\,,
\end{align}
where $\eta$ is a scalar meromorphic function on $\Gamma$ uniquely determined by the choice of ${\bf a} \in (\Complex^*)^n$. Explicitly, one has
\begin{align}\label{etaDef}\eta=\sum\nolimits_{i=1}^n a_i(\phi(\lambda, L)_+h_{L})^i \,.
\end{align}
Further, from equation \eqref{heq} it follows that the coordinates of poles of $h_{L}$ satisfy the equation
\begin{align}\label{poleEvolution}
\diffXp{t} z(P) = \Res_P( \eta \diff z )\,,
\end{align}
where $z$ is any local coordinate near the pole $P$. In particular, for sufficiently small times the number of poles remains unchanged, the poles do not collide and stay away from the points of $\Lambda'$. Therefore, if 
${\bf a} \in (\Complex^*)^n$ is chosen in such a way that the divisor $D_{L_0}$ is simple and $\Lambda'$-regular, the same is true for $D_{L}$ for any $L \in O(L_0)$ sufficiently close to $L_0$.
In particular, the multidegree of $D_{L}$ is the same for any $L \in O(L_0)$ (note, however, that in contrast to the smooth case, for singular curves it is not true that the multidegree of $D_{L}$ is the same for any $L$ in the \textit{fiber} of $L_0$).
\par
Now, taking the $\Lambda'$-linear equivalence class of the divisor $D_{L}$, we obtain a map
\begin{align*}
\Psi \colon O(L_0) &\to \Pic_\gamma(\Gamma,\Lambda')\\
L &\mapsto \left[D_{L}\right]_{\Lambda'}\,,
\end{align*}
 where $\gamma$ is some fixed multidegree on $\Gamma$. 
 Note that the $\Lambda'$-equivalence class of $D_{L}$ (in contrast to its $\Lambda$-linear equivalence class) does depend on the choice of ${\bf a} \in \Complex^n$, so the mapping $\Psi$ is not uniquely defined. However, the mappings $\Psi$ associated with different choices of the vector ${\bf a}$ differ from each other by a shift in the Picard group.\par
\smallskip
\textit{Step 2. Linearization of flows.} 
 Now, we prove formula \eqref{velocity}. Assume that $L = L(t)$ evolves according to equation~\eqref{loopI}. Then from~\eqref{poleEvolution} it follows that
\begin{align*}
\diffXp{t} \int_{D_{L(0)}}^{D_{L(t)}}\omega \,=\,\sum\nolimits_{P \in D_{L(t)}} \Res_{P}\,\eta\omega \,,
\end{align*}
where $\omega$ is any meromorphic differential on $\Gamma$.
If, moreover, one has $\omega \in \Omega^1(\Gamma,\Lambda')$, then from formula~\eqref{etaDef} it follows that the form $\eta\omega$ may only have poles at points of $\mathrm{supp}(\Lambda') = \mathrm{supp}(\Lambda) \sqcup  \Gamma_\infty$ and points of $D_{L}$. So, by the residue theorem, we have 
\begin{align}\label{sumOfResidues}
\sum\nolimits_{P \in D_{L(t)}}\Res_{P}\,\eta\omega = -\sum\nolimits_{i=1}^n \Res_{\infty_i}\, \eta\omega\, -\, \sum\nolimits_{\{Q_i^{\pm}\} \in \Lambda}\left(\Res_{Q_i^+}\, \eta\omega + \Res_{Q_i^-}\, \eta\omega\right).
\end{align}
Further, recall that the set $\Lambda$ consists of preimages on non-essential double points of the spectral curve. Therefore, by definition of non-essential double points, we have $h_{L}^i \in \pazocal M(\Gamma,\Lambda)$, and thus $\eta \in \pazocal M(\Gamma,\Lambda)$. Taking into account that $\omega \in \Omega^1(\Gamma,\Lambda')$, it follows that the latter sum in \eqref{sumOfResidues} vanishes. To compute the former sum, write
\begin{align*}
 \Res_{\infty_i}\, \eta\omega = \sum\nolimits_{j=1}^n    a_j \Res_{\infty_i} \left( \phi(\lambda, L)_+\,h_{L} \right)^j \omega\,.
\end{align*}
Note that $\mathrm{ord}_{\infty_i} \!\left(  (\phi(\lambda, L) - \phi(\lambda, L)_{+})h_{L} \right)^j \geq 1$, and $\mathrm{ord}_{\infty_i} \omega \geq -1$, so
$$
\Res_{\infty_i} \left( \phi(\lambda, L)_+\,h_{L} \right)^j \omega = \Res_{\infty_i} \left( \phi(\lambda, L)\,h_{L} \right)^j \omega = \Res_{\infty_i} \, \phi(\lambda, \mu)\,h_{L}^j \omega\,,
$$
and
$$
\diffXp{t} \int_{D_{L(0)}}^{D_{L(t)}}\omega \,=\, -\sum\nolimits_{i=1}^n \Res_{\infty_i}\, \eta\omega = - \sum\nolimits_{i=1}^n \sum\nolimits_{j=1}^n   a_j\Res_{\infty_i} \phi(\lambda, \mu)\,h_{L}^j \omega =- \sum\nolimits_{i=1}^n  \,\Res_{\infty_i}\,  \phi(\lambda,\mu) \omega \,,
$$
ending the proof.\par
Note that from formulas~\eqref{loopI} it follows that for a suitable choice of coordinates on $O(L_0)$ the eigenvector mapping $\Psi$ is \textit{linear}. In particular, $\Psi$ is holomorphic, and to prove that it is locally biholomorphic, it suffices to show that $\Psi$ is injective, and that its differential is surjective.\par \smallskip
\textit{Step 3. The eigenvector map is injective.}
We begin with studying the local properties of the vector-function $h_{L}$ at ramification points and nodes of the spectral curve. For a given $L \in O(L_0)$, consider the set $\Gamma_0 := \Gamma \, \setminus \, (\mathrm{supp}(D_{L}) \sqcup \Gamma_\infty)$.
\begin{proposition}\label{hProp} For any $L \in O(L_0)$, one has the following:
\begin{enumerate}
\item Let $P \in \Gamma_0$ be a simple (i.e., multiplicity two) ramification point of $\lambda$. Then the matrix $L(\lambda(P))$ has a $2\times 2$ Jordan block with eigenvalue $\mu(P)$, eigenvector $h_{L}(P)$, and generalized eigenvector $\partial_z h_{L}(P)$, where $h_L$ is given by \eqref{hFormula}, $\partial_z = \diff / \diff z$, and $z$ is any local parameter near $P$.
\item Let $P^+, P^- \in \Gamma_0$ be such that $\pi(P^+) = \pi(P^-)\in \pazocal E$. Assume also that $P^-$ and $P^+$ are not ramification points of $\lambda$. Then the eigenvectors $h_{L}(P^+)$, $h_{L}(P^-)$ are linearly independent.
\item Let $P^+, P^- \in \Gamma_0$ be such that $\pi(P^+) = \pi(P^-)\in \pazocal{\bar E}$. Assume also that $P^-$ and $P^+$ are not ramification points of $\lambda$. Then the matrix ${L}(\lambda(P))$ has a $2\times 2$ Jordan block with eigenvalue $\mu(P^-) = \mu(P^+)$, eigenvector $h_{L}(P^+) = h_{L}(P^-)$, and generalized eigenvector $\partial_\lambda h_{L} (P^+) - \partial_\lambda  h_{L} (P^-)$. 
\end{enumerate}
\end{proposition}
\begin{proof}
The first statement is well-known (see, e.g., Chapter 5.2 of~\cite{babelon}). To prove the second statement, differentiate the identity $(L(\lambda(P)) - \mu(P)\E)h_{L}(P) = 0$ with respect to the local parameter $\lambda$ at the points $P^-$ and $P^+$, and subtract the resulting equations. 
Then, assuming that the vectors $h_L(P^-)$ and $h_L(P^+)$ are dependent (and, therefore, equal), one obtains
\begin{align}\label{JordanCond}
\begin{aligned}
\left(L(\lambda(P^\pm)) - \mu(P^\pm)\E\right)\left(\partial_{\lambda}h_L(P^+) - \partial_{\lambda}h_L(P^-)\right) = \left(\partial_\lambda \mu(P^+)- \partial_\lambda \mu(P^-)\right)h_L(P^\pm)\,.
\end{aligned}
\end{align}
Since the point $\pi(P^-) = \pi(P^+)$ is a node, the right-hand side of the latter equation is not zero. Therefore, the matrix $L(\lambda(P^\pm))$ has a Jordan block with the eigenvalue $\mu(P^\pm)$. But this contradicts the assumption $\pi(P^-) = \pi(P^+) \in \pazocal E$. Thus, the second statement is proved. The proof of the third statement is analogous and follows from equation~\eqref{JordanCond}.
\end{proof}

 Now, for a $\Lambda$-regular divisor $D \in \Div(\Gamma, \Lambda)$, let
 $$
\mathcal {L}(D,\Lambda) := \{ f \in \pazocal M(\Gamma, \Lambda) \mid \mathrm{ord}_P\,f \geq -D(P) \,\,\forall\,\, P \in \Gamma \}\,.
$$
Here $\mathrm{ord}_P \,f$ denotes the order of $f$ at $P$, and, by definition, we set $\mathrm{ord}_P \, := +\infty$ if $f \equiv 0$ at the connected component of $\Gamma$ containing $P$. 

\begin{proposition}\label{nonSpec1}
Let $D_\infty := \sum\nolimits_{i=1}^n\infty_i$. Then, for any $L \in O(L_0)$, one has $
 \mathcal {L}\left(D_L- D_\infty, \Lambda\right) = 0.
$
\end{proposition}
\begin{proof}
Let $ S \colon \mathcal{L}(D_L -D_\infty, \Lambda) \to \mathcal{L}(D_L, \Lambda) $ and $ T \colon \mathcal{L}(D_L, \Lambda) \to \mathcal{L}(D_L - D_\infty, \Lambda) $ be linear maps given by $$ S( f) := \lambda f, \quad T(f) := f- \sum\nolimits_{i=1}^n f(\infty_i)h_L^i,$$
where $h^i_L$ are components of the eigenvector $h_L$. Assume that $\dim \mathcal{L}(D_L -D_\infty, \Lambda) > 0$. Then the operator $$T \circ S \colon  \mathcal{L}(D_L -D_\infty, \Lambda) \to  \mathcal{L}(D_L -D_\infty, \Lambda) $$ must have an eigenvector $g \in  \mathcal{L}(D_L -D_\infty, \Lambda)  $. Denote the corresponding eigenvalue by $\alpha$. We have \begin{align}\label{badEquation}(\lambda - \alpha) g = \sum\nolimits_{i=1}^n c_i h_L^i\end{align} where $c_1, \dots, c_n \in \Complex$ are some constants. \par
Let us show that~\eqref{badEquation} cannot hold. To begin with, note that without loss of generality we may assume that $\lambda \neq \alpha$ at points of $D_L$. If not, we change the vector ${\bf a}$ in formula~\eqref{hFormula}. This leaves the spectrum of the operator $T \circ S$ invariant (it is easy to see that operators $T \circ S$ corresponding to different choices of ${\bf a} \in (\Complex^*)^n$ are conjugate to each other) and moves the divisor $D_L$ away from the line $\lambda = \alpha$. Further, for simplicity, assume that the spectral curve $C$ satisfies the following genericity assumption: the geometric number of intersection points of $C$ with the line $\lambda = \alpha$ is at least $n-1$ (the general case is analogous but computationally more elaborate).
Under this genericity assumption, the divisor $(\lambda)_\alpha := \{\lambda =\alpha \}$ on $\Gamma$ can be of one of the following types:
\begin{enumerate}\item $(\lambda)_\alpha = \sum\nolimits_{i=1}^n P_i$, where  $P_i$'s are pairwise distinct and project to smooth points of the spectral curve.
 \item  $(\lambda)_\alpha = 2P_{n-1} + \sum\nolimits_{i=1}^{n-2} P_i$, where  $P_i$'s are same as above.
 
  \item  $(\lambda)_\alpha = P^+ + P^- + \sum\nolimits_{i=1}^{n-2} P_i$, where  $P_i$'s are same as above, and $P^\pm$ are the preimages of an essential double point $P \in \pazocal E$.
   \item  $(\lambda)_\alpha = P^+ + P^- + \sum\nolimits_{i=1}^{n-2} P_i$, where $P_i$'s are same as above, and $P^\pm$ are the preimages of a non-essential double point $P \in \bar{\pazocal E}$.

\end{enumerate}
To prove that equation~\eqref{badEquation} cannot hold, we consider each of these four cases separately. 
 In the first case, from~\eqref{badEquation} it follows that the vectors $h_L(P_1), \dots, h_L(P_n)$ are linearly dependent. But this is not possible, as the latter vectors are the eigenvectors of $L(\alpha)$ corresponding to distinct eigenvalues. Similarly, in the second case we obtain the linear dependence of $h_L(P_1), \dots, h_L(P_{n-1}),  \partial_z h_L(P_{n-1})$, but this contradicts the first statement of Proposition~\ref{hProp}.
Further, notice that the third case is analogous to the first one in view of the second statement  of Proposition~\ref{hProp}. Finally, in the fourth case it follows from~\eqref{badEquation} that
\begin{align*}
 g(P^+) - g(P^-) =  \sum\nolimits_{i=1}^n c_i \left(\partial_\lambda{ h_L^i }(P^{+}) -  {\partial_\lambda  h_L^i } (P^{-}) \right).
\end{align*}
At the same time, we have $\g \in \pazocal M(\Gamma, \Lambda)$, so $ g(P^+) = g(P^-) $. Therefore, the vectors
$$
h_L(P_1),\, \dots, h_L(P_{n-2}), h_L(P^+),\partial_\lambda{ h_L }(P^{+}) -  {\partial_\lambda  h_L }(P^{-})$$
are linearly dependent. But this contradicts the last statement of Proposition~\ref{hProp}. Thus, Proposition~\ref{nonSpec1} is proved.


%
\end{proof}
\begin{proposition}\label{nonSpec2}
For any $L \in O(L_0)$, and any $i = 1, \dots, n$, one has $
\dim \mathcal{L}(D_L -D_\infty + \infty_i, \Lambda) \leq 1$.
\end{proposition}
\begin{proof}
Let $f,g $ be non-zero elements of   $\mathcal{L}(D_L -D_\infty + \infty_i, \Lambda) $. Then $f(\infty_i)g - g(\infty_i)f \in \mathcal {L}(D_L -D_\infty, \Lambda) $. So, from Proposition~\ref{nonSpec1} it follows that $f(\infty_i)g - g(\infty_i)f = 0$, i.e. $f$ and $g$ are linearly dependent (note that $f(\infty_i) \neq 0$ because otherwise we have $f \in \mathcal{L}(D_L -D_\infty , \Lambda)$, and thus $f = 0$). The result follows. 
\end{proof}

\begin{proposition}\label{nonSpec3}
For any $L \in O(L_0)$, one has $ \mathcal{L}(D_L, \Lambda') = \Complex$ (where $\Complex$ denotes constant functions).
\end{proposition}
\begin{proof}
Since the divisor $D_L$ is effective, we have $\Complex \subset  \mathcal{L}(D_L, \Lambda')$, and it suffices to prove the opposite inclusion.
Let $f \in \mathcal{L}(D_L, \Lambda') $. Then $ f - f(\infty_1) \in  \mathcal{L}(D_L - D_\infty, \Lambda)$. So, by Proposition \ref{nonSpec2}, we have $f - f(\infty_1) = 0$, and $f = f(\infty_1) = \const$, as desired.
\end{proof}
Now, we are in a position to prove that the mapping $\Psi$ is injective. Assume that $\Psi(L_{1}) =\Psi( L_{2})$, where $L_1, L_2 \in O(L_0)$. 
Notice that, by construction of $h_L$ and $D_L$, we have
$$h_{L_1}^i \in \mathcal{L}(D_{L_1} -D_\infty + \infty_i, \Lambda), \quad  h_{L_2}^i \in \mathcal{L}(D_{L_2} -D_\infty + \infty_i, \Lambda).$$ 
At the same time, the equality $\Psi(L_{1}) =\Psi( L_{2})$ means that $[D_{L_1}]_{\Lambda'} =[D_{L_2}]_{\Lambda'} , $ so $D_{L_2} - D_{L_1} = (f)$ where $f \in \mathcal{L}(D_{L_1}, \Lambda')$. But from Proposition~\ref{nonSpec3} it follows that $f = \const$, therefore $D_{L_2} = D_{L_1}$. So, we have
 $$h_{L_1}^i, h_{L_2}^i \in \mathcal{L}(D_{L_1} -D_\infty + \infty_i, \Lambda),$$ and thus by Proposition~\ref{nonSpec2} the functions $h_{L_1}^i, h_{L_2}^i $ are linearly dependent. At the same time,  we have $h_{L_1}^i(\infty_i) = h_{L_2}^i(\infty_i) =1 / a_i $, so  $h_{L_1}^i = h_{L_2}^i $ and thus $h_{L_1} = h_{L_2} $. Therefore, for each $\alpha \in \Complex$, the matrices $L_{1}(\alpha)$ and $L_{2}(\alpha)$ have the same eigenvalues and the same eigenvectors, and thus coincide. So, injectivity is proved.

\par
\smallskip
\textit{Step 4. The differential of the eigenvector map is surjective.}
To prove that the differential of $ \Psi$ is surjective, we need the following proposition.
\begin{proposition}\label{bfNonDeg}
The pairing $ \InnerProduct_\infty$ is non-degenerate from the right, i.e. for any $\omega \in  \Omega_m^1(\Gamma)$, there exists $\phi \in \Complex[\lambda^{\pm 1}, \mu]$ such that $\langle \phi, \omega \rangle_\infty \neq 0$.
\end{proposition}
\begin{proof}
 Let
$
m := \min_{i} \mathrm{ord}_{\infty_i}\,\omega
$, and let $f(k) := m + 1 - dk$.
Then, for any integer $k \geq 0$, we have
 $$ \mathrm{ord}_{\infty_i}\,\mu^k\lambda^{f(k)}\omega \geq -1\,.$$
  Further, using that $\mu$ at $\infty_i$ is given by $\mu = j_i \lambda^d + \dots$, where $j_1, \dots, j_n$ are the eigenvalues of $J$, we have
$$
\langle \mu^k\lambda^{f(l)}, \omega \rangle_\infty =  \sum\nolimits_{i=1}^n  j_i^k \,\Res_{\infty_i}\,\lambda^{m+1}\omega\,.
$$
Assume that the latter expression vanishes for any integer $k\geq 0$. Then, since the numbers $j_1, \dots, j_n$ are all distinct, it follows that $\Res_{\infty_i}\,\lambda^{m+1}\omega = 0$ for each $i = 1, \dots, n$. But this contradicts the choice of $m$. Thus, the form $ \InnerProduct_\infty$ is indeed non-degenerate from the right, q.e.d.
%
\end{proof}
Now, we finish the proof of Lemma~\ref{lemma1}. Clearly, being restricted to $\Complex[\lambda^{\pm 1}, \mu] \times \Omega^1(\Gamma, \Lambda')$, the pairing $\InnerProduct_\infty$ remains non-degenerate from the right, that is to say that the sequence 
$$0 \rightarrow \Omega^1(\Gamma, \Lambda') \xrightarrow{i_\infty} \Complex[\lambda^{\pm 1}, \mu]^*,$$
where $i_\infty$ is given by $\omega \mapsto \langle \,\,, \omega\rangle_\infty $, is exact. Since the space $ \Omega^1(\Gamma, \Lambda')$ is finite-dimensional, the dual sequence
$$0 \leftarrow \Omega^1(\Gamma, \Lambda')^* \xleftarrow{i^*_\infty} \Complex[\lambda^{\pm 1}, \mu]$$
is exact as well. 
But the latter is equivalent to saying that the mapping $i_\infty^*$ given by $\phi \mapsto \langle \phi, \,\rangle_\infty$ is surjective. So, vectors of the form~\eqref{velocity} span the tangent space $\Omega^1(\Gamma, \Lambda')^*$ to $\Pic(\Gamma, \Lambda')$, which means that $\diff\Psi$ is surjective. Thus, Lemma \ref{lemma1} is proved.

\end{proof}

\medskip
\subsection{Eigenvalues of linearized flows via residues}\label{sec:evvr}
It turns out that when a vector field of the form~\eqref{loopI} vanishes at a point  $L \in  \pazocal P_d^J (\gl_n)$, the eigenvalues of its linearization admit a representation similar to~\eqref{velocity}. This is a key result for proving that singularities corresponding to nodal curves are non-degenerate.\par
As in Section \ref{sec:lf}, let $C$ be a nodal curve, $\pazocal E =  \{P_1, \dots, P_k\} \subset \Sing C$ be a subset of its nodes, and let $\pazocal T_{C, \pazocal E}$ be the set given by~\eqref{ticieps}. Let also $\Gamma$ be the non-singular compact model of $C$. Then each point $P_j$ corresponds to two points $P_j^+, P_j^- \in \Gamma$. Let also $\Lambda'$ be given by \eqref{LambdaPrime}. Consider a differential $\omega_j$ on $\Gamma$ which is $\Lambda'$-regular everywhere except the points $P_j^+, P_j^-$, and such that
\begin{align}\label{resCond}
\Res_{P_j^+} \omega_j =  1, \quad \Res_{P_j^-} \omega_j =  -1\,.
\end{align}
Obviously, such a differential $\omega_j$ exists and is unique modulo $\Lambda'$-regular differentials.
\begin{lemma}\label{lemma2}
Let $L \in \pazocal T_{C, \pazocal E}$. Then there exist non-zero $X_1^\pm, \dots, X_k^\pm \in \T_L^*\pazocal P_d^J (\gl_n)$ such that for any vector field of the form~\eqref{loopI} vanishing at $L$ we have
\begin{align}
\label{eigenvalues}
D_\phi^* X_j^\pm =  \pm\langle \phi, \omega_j \rangle_\infty X_j^\pm,
\end{align}
where $D_\phi \in \End( \T_L\pazocal P_d^J (\gl_n) ) $ is the linearization of the vector field~\eqref{loopI} at the point $L$, $D_\phi^*  \in \End( \T^*_L\pazocal P_d^J (\gl_n) )$ is the dual of $D_\phi$, and the form $\InnerProduct_\infty$ is given by \eqref{inftyForm}.
\end{lemma}
\begin{remark}
Note that for a vector field~\eqref{loopI} vanishing at $L$, the right-hand side of~\eqref{eigenvalues} does not depend on the choice of the form $\omega_j$. Indeed, $\omega_j$ is unique up to a $\Lambda'$-regular differential, and for $\Lambda'$-regular differentials the right-hand side of~\eqref{eigenvalues}  vanishes by Lemma~\ref{lemma1}.
\end{remark}
\begin{proof}[Proof of Lemma \ref{lemma2}]
Consider the map $A_\phi \colon \pazocal P_d^J (\gl_n) \to  \gl_n \otimes \Complex[\lambda]$ given by $A_\phi( L) := \phi(\lambda, L)_+$. In terms of this map, the right-hand side of~\eqref{loopI} reads as $[L, A_\phi(L)]$. Assume that the latter commutator vanishes for a certain $L \in  \pazocal P_d^J (\gl_n) $, and let $A: = A_\phi(L)$. Then the linearization of~\eqref{loopI} at $L$ is a linear operator $D_\phi \in \End(\T_L\pazocal P_d^J (\gl_n) )$ given by \begin{align*}D_\phi(X) &= [X,A] + [L,Y]\,,
\end{align*}
where $Y$ is the image of $X$ under the differential of the map $A_\phi$ at the point $L$.
Further, consider the map
$
R \colon \Complex \times \gl_n  \to \T^*_L\pazocal P_d^J (\gl_n)
$
given by
$
\langle R(z, Z), X \rangle := \Tr ZX(z)
$. 
Then
\begin{align*}
\langle D_\phi^*(R(z, Z)), X \rangle = \langle R(z, Z), \,&D_\phi(X) \rangle=   \Tr Z[X(z),A(z)] + \Tr Z[L(z),Y(z)]  \\ &= \Tr [A(z),Z]X(z) + \Tr [Z,L(z)]Y(z).
\end{align*}
Assuming that the matrix $Z$ belongs to the centralizer $ C({L(z)})$ of $L(z)$, we have
\begin{align}\label{coFormula}
D_\phi^*(R(z, Z)) = R(z, [A(z),Z])\,.
\end{align}
From the latter formula it follows that for every $z \in \Complex$ the subspace $R(z,  C({L(z)})) \subset  \T^*_L\pazocal P_d^J (\gl_n)$ is invariant with respect to the operator $D_\phi^*$. 
Using this observation, we explicitly construct the eigenvectors of $D_\phi^*$.  
\par
Formula~\eqref{hFormula} defines a non-vanishing meromorphic vector-function $h_L \colon \Gamma \to \Complex^{n}$ on the Riemann surface $\Gamma$ such that
$
Lh_L = \mu h_L.
$  In a similar way, one constructs a non-vanishing meromorphic vector-function $\xi_L$ such that
$
L^*\xi_L = \mu \xi
$, where $L^*$ is the dual of $L$ (i.e., the transposed matrix). Notice that since $[L,A] = 0$, there exists a meromorphic function $\nu$ on $\Gamma$ such that
\begin{align}\label{AAction}
Ah = \nu h, \quad A^*\xi = \nu \xi\,.
\end{align}

Now, let $P_j \in \pazocal E$ be an essential double point of the spectral curve, and let $P_j^+$, $P_j^-$ be the corresponding points on the Riemann surface $\Gamma$.  
Consider $X_j^+, X_j^- \in \T^*_L\pazocal P_d^J (\gl_n)$ given by
$$
X_j^+ =R(\lambda(P_j), h(P_j^-) \otimes \xi(P_j^+)), \quad X_j^- := R(\lambda(P_j) ,h(P_j^+) \otimes \xi(P_j^-)).
$$
(Here we assume that the normalization of $h$ and $\xi$ is chosen in such a way that these functions are finite at the points $P_j^\pm$.)
Then from formulas~\eqref{coFormula} and~\eqref{AAction} it follows that
$$
D_\phi^*X_j^\pm =  \pm(\nu(P_j^-) - \nu(P_j^+))X_j^\pm\,.
$$
Now, to complete the proof of Lemma~\ref{lemma2}, it remains to show that 
$$
\nu(P_j^-) - \nu(P_j^+) =  \langle \phi, \omega_j \rangle_\infty \,,
$$
where $\omega_j$ is a differential on $\Gamma$ which is $\Lambda'$-regular everywhere except the points $P_j^\pm$, and such that
$\Res_{P_j^\pm}\, \omega_j = \pm 1\,.$
We have 
\begin{align}\label{evComp}
\begin{aligned}
\nu(P_j^-) - \nu(P_j^+) &= -\Res_{P_j^+}\, \nu\omega_j - \Res_{P_j^-}\, \nu\omega_j  \\ &=\,\sum\nolimits_{i=1}^n \Res_{\infty_i}\, \nu \omega_j \,+ \!\sum\nolimits_{\{Q_i^\pm\} \in \Lambda}\left(\Res_{Q_i^+}\, \nu\omega_j + \Res_{Q_i^-}\, \nu\omega_j\right),
\end{aligned}
 \end{align}
 where $\Lambda$ is given by \eqref{Lambda}. Now the key observation is that since for any $\{Q_i^\pm\} \in \Lambda$ the corresponding double point of the spectral curve is non-essential, we have $h(Q_i^+) = h(Q_i^-)$, and, therefore, the function $\nu$ is $\Lambda$-regular. Thus, the latter sum in \eqref{evComp} vanishes. To compute the former sum, notice that
$$
\phi(\lambda, L)_+\,h = \nu h, \quad \phi(\lambda, L)h = \phi h\,,
$$
so
$$
(\phi(\lambda, L) -  \phi(\lambda, L)_+)h = (\phi - \nu) h\,,
$$
and it follows that $\mathrm{ord}_{\infty_i} \left(\phi- \nu\right) \geq 1$. Therefore, we have
$$
\nu(P_j^-) - \nu(P_j^+)  = \sum\nolimits_{i=1}^n \Res_{\infty_i}\, \nu\omega_j =  \sum\nolimits_{i=1}^n \Res_{\infty_i}\, \phi\omega_j =  \langle \phi, \omega_j \rangle_\infty \,,
$$
ending the proof.
\end{proof}

\section{Proofs of the main theorems}


\subsection{Rank and non-degeneracy: proof of Theorem~\ref{thm1}}\label{sec:pt1}

From Lemma~\ref{lemma1} and Proposition \ref{genusFormula} it follows that
\begin{align}
\rank L = \dim O(L) = \dim \Pic(\Gamma, \Lambda') =  \frac{1}{2}dn(n-1) - |\pazocal E_L|\,.
\end{align}
Now, the desired formula
$
\corank L =  |\pazocal E_L|\,
$
follows from
\begin{proposition}\label{leaves}
All symplectic leaves of the Poisson bracket~\eqref{PoissonI} have dimension $dn(n-1)$.
\end{proposition}
\begin{proof}
The corresponding Poisson tensor $ \Pi$, considered as a map
$ \T^*_L\pazocal P_d^J (\gl_n) \to \T_L\pazocal P_d^J (\gl_n)
$, reads
$
\Pi(Y) = [L,Y]_+,
$
where, as in Section~\ref{sec:ismp}, we identify the tangent space to $\T_L\pazocal P_d^J (\gl_n)$ with the space of matrix polynomials of the form \eqref{tangent}, while the cotangent space  $\T^*_L\pazocal P_d^J (\gl_n)$ is identified with the space of Laurent polynomials of the form \eqref{cotangent}. 
%
 To find the dimension of the symplectic leaf containing the point $L$, we describe the kernel of the Poisson tensor $\Pi$ at $L$, i.e. solve the equation \begin{align}\label{kerEqn}[L,Y]_+ = 0\,.\end{align} To begin with, observe that for every polynomial $\phi(\lambda, \mu)$ of degree at most $d-1$ in $\lambda$ and at most $n-1$ in $\mu$, the matrix
\begin{align}
\label{kerPoisson}
Y(\lambda) = \lambda^{-d}\phi(\lambda, \lambda^{-d}L)_+\,
\end{align}
is a solution of equation~\eqref{kerEqn}, and all these solutions are pairwise distinct. Thus, the dimension of the kernel of $\Pi$ is at least $dn$. On the other hand, it is easy to see that equation~\eqref{kerEqn} cannot have diagonal-free solutions, so by the dimension argument it must be that $\dim \Ker \Pi = dn$, i.e. all solutions of equation~\eqref{kerEqn} have the form~\eqref{kerPoisson}. Thus, the dimension of the symplectic leaf  containing $L$ is equal to $$\dim \Imm \Pi =\dim \pazocal P_d^J (\gl_n) - \dim \Ker \Pi = dn(n-1)\,,$$ as desired.
%
%
%
\end{proof}
%
Now, we prove that under the assumptions of Theorem~\ref{thm1} the singular point $L$ is non-degenerate. In the Poisson setting, it is convenient to prove non-degeneracy using the following criterion, which is a reformulation of Corollary 4.2 of~\cite{bolsinov2014singularities}:
 \begin{lemma}\label{NDC}
Let $\manifold$ be a Poisson manifold, and let $ S \subset \manifold$ be a generic symplectic leaf. Let also $\F \subset \pazocal O(M)$ be such that the restriction of $\F$ to $S$ is an integrable system. Further, let $x \in { S}$ be a singular point of corank $k$ for the restriction of $\F$ to $S$. Consider the vector space  $\F_x = \{   H \in \F \mid \mathrm{X}_H(x) = 0\}$. Then $x$ is non-degenerate if and only if there exist $2k$ non-zero cotangent vectors $\xi_1^\pm, \dots, \xi_{k}^\pm \in \T^*_x \manifold$ and $k$ linearly independent functionals $\alpha_1, \dots, \alpha_k \in \F_x^*$ such that for each $H \in \F_x$, we have
\begin{align}\label{DStarEV}
(\mathrm{DX}_H)^*\xi_j^\pm = \pm \alpha_j(H)\xi_j^\pm\,.
\end{align}
Furthermore, in the real case the type of a non-degenerate point $x$ is $(k_e,k_h,k_f)$ where $k_e$ is the number of pure imaginary $\alpha_i$'s, $k_h$ is the number of real $\alpha_i$'s, and $k_f$ is the number of pairs of complex conjugate $\alpha_i$'s.

\end{lemma}

To apply this result in our setting, notice that by Lemma \ref{lemma1} the space of vector fields of the form~\eqref{loopI} vanishing at $L$ can be identified with the orthogonal complement $\Omega(\Gamma, \Lambda')^\bot$ to $\Omega(\Gamma, \Lambda')$ with respect to the $\InnerProduct_\infty$ pairing. 
Furthermore, by Lemma \ref{lemma2} there exist meromorphic differentials $\omega_1, \dots, \omega_k \in \Omega^1_m(\Gamma)$, and non-zero $X_1^\pm, \dots, X_k^\pm \in \T_L^*\pazocal P_d^J (\gl_n)$ such that for any $\phi \in \Omega(\Gamma, \Lambda')^\bot$, we have
\begin{align*}
D_\phi^* X_j^\pm = \pm \langle \phi, \omega_j \rangle X_j^\pm.
\end{align*}
Further, the same argument as in the proof of Lemma~\ref{lemma1} shows that the mapping $$ \Complex[\lambda^{\pm 1}, \mu] \to (\Omega^1(\Gamma, \Lambda') \oplus \mathrm{span} \langle \omega_1, \dots, \omega_k \rangle)^* $$
induced by the pairing $\InnerProduct_\infty$ is surjective. But from this it follows that the mapping
$$\Omega(\Gamma, \Lambda')^\bot \to (\mathrm{span} \langle \omega_1, \dots, \omega_k \rangle)^* $$
induced by the same pairing is surjective as well. The latter is equivalent to saying that functionals of the form $ \langle \,\, , \omega_j \rangle_\infty$ are linearly independent on $\Omega(\Gamma, \Lambda')^\bot$. Now, non-degeneracy of $L$ follows from Lemma \ref{NDC}.
 Thus, Theorem~\ref{thm1} is proved.
\medskip
\subsection{Williamson types: proof of Theorem~\ref{thm2}}\label{sec:pt2}
 By Lemma~\ref{NDC}, the type of $L$ is given by the number of pure imaginary, the number of real, and the number of pairs of complex conjugate functionals $\phi \mapsto  \langle \phi, \omega_j \rangle_\infty$. To find these numbers, consider the antiholomorphic involution $ \Gamma \to \Gamma$ induced by the involution $\tau$ on the spectral curve. We denote this involution by the same letter $\tau$. In what follows, we will need the following straightforward result.
 \begin{proposition}
 Let $\Gamma$ be a Riemann surface endowed with an antiholomorphic involution $\tau$, and let $\omega$ be a meromorphic differential on $\Gamma$. Let also $P \in \Gamma$. Then
\begin{align}\label{conjres}
\overline{\vphantom{ \tau^*}\Res_{ P}\, \omega} = \Res_{ \tau(P)}\, \overline{ \tau^*\omega}\,.
\end{align}
\end{proposition}
Now, assume that $P_j$ is an $L$-essential isolated point of the spectral curve, and  let $P_j^\pm$ be the corresponding points on the Riemann surface $\Gamma$. Then the involution $ \tau$ interchanges $P_j^+$ with $P_j^-$, and from formula~\eqref{conjres} it follows that $ \overline{ \tau^*\omega_j} = -\omega_j$ modulo a $\Lambda'$-regular differential. So, we have
$$
\overline{ \langle \phi, \omega_j \rangle}_\infty =  \sum\nolimits_{i=1}^n \Res_{ \tau(\infty_i)}\, \overline{\tau^*(\phi\omega_j)} = -\sum\nolimits_{i=1}^n \Res_{\infty_i}\, \overline{\tau^*\phi}\omega_j
$$
where we used that the set $\{ \infty_1, \dots, \infty_n\}$ is invariant under $\tau$. Further, by Proposition \ref{propInv}, for real forms of non-compact type, we have $\overline{\tau^*\phi }= \phi$, so
$$
\overline{ \langle \phi, \omega_j \rangle}_\infty = - { \langle \phi, \omega_j\rangle}_\infty\,,
$$
which means that each essential isolated point gives rise to an elliptic component. On the other hand, by the same Proposition \ref{propInv}, in the compact case we have $\overline{\tau^*\phi }= -\phi$, so
$$
\overline{ \langle \phi, \omega_j \rangle}_\infty = { \langle \phi, \omega_j\rangle}_\infty\,,
$$
which means that in this case each essential isolated point corresponds to a hyperbolic component.\par
 Analogously, if $P_j$ is an essential self-intersection point, then $\overline{ \tau^*\omega_j} = \omega_j$, therefore $P_j$ gives rise to a hyperbolic component for real forms of non-compact type and to an elliptic component for real forms of compact type.\par
 Finally, it is easy to see that a pair of essential double points interchanged by $\tau$ always leads to a focus-focus component. Thus, Theorem~\ref{thm2} is proved.

\bigskip
\section{Examples}\label{sec:examples}
\subsection{Shift of argument systems}\label{sec:sas}
As an example, we consider polynomial matrix systems on degree one polynomials in $\pazocal P_1^J (\g)$, where $\g$ is a real form of $\gl_n(\Complex)$. 
These systems are known as \textit{shift of argument} systems. They were introduced in \cite{MF} as a generalization of the multidimensional free rigid body equation considered in \cite{Manakov}.
\par
Fix $\g$ and $J \in \g$, and consider rank zero singular points of the polynomial matrix system on $\pazocal P_1^J (\g)$. These points are given by matrix pencils of the form $L = L_0 + \lambda J$, where $L_0$ lies in the Cartan subalgebra $\mathfrak t$ containing $J$ (see Example \ref{gl2ex}). Since $J$ can be diagonalized over $\Complex$, it follows that the corresponding spectral curve is the union of $n$ straight lines $l_1, \dots, l_n$. Let us assume that these lines are in general position, so that $L$ is a non-degenerate singular point. \par To compute the type of $L$, notice that the action of the antiholomorphic involution $\tau$ on the set  $\{l_1, \dots, l_n\}$ depends only on (the conjugacy class of) the Cartan subalgebra $\mathfrak t$. Namely, consider the eigenvalues $\mu_1, \dots, \mu_n$ of an element of $\mathfrak t$ as linear functionals $\mu_i \colon \mathfrak t \to \Complex$. Then there is an involution $\hat \tau$ on the set $\mu_1, \dots, \mu_n$ given by $\hat \tau \colon \mu \mapsto \bar \mu$ for real forms of non-compact type, and by $\hat \tau \colon \mu \mapsto -\bar \mu$ for real forms of compact type. Furthermore, from the definition of the spectral curve it follows that the action of the involution $\tau$ on the set $\{l_1, \dots, l_n\}$ is the same as the action of the involution $\hat \tau$ on the set $\{\mu_1, \dots, \mu_n\}$.\par
Now, to each Cartan subalgebra $\mathfrak t \subset \mathfrak h$ we associate a complete graph $G$ endowed with an involution~$\hat \tau$. The vertices of this graph are eigenvalues $\mu_1, \dots, \mu_n \in \mathfrak t^* \otimes \Complex$, every pair of vertices is joined by an edge, and the involution $\hat \tau$ is the same as above. Clearly, this graph is uniquely determined by the conjugacy class of $\mathfrak t$, and vice versa, the conjugacy class is uniquely determined by the graph (however, not all graphs are possible for a given real form). 
 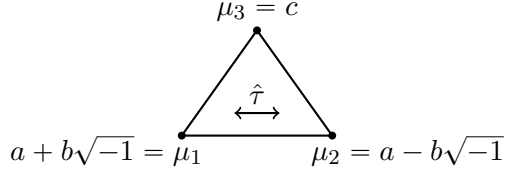
\begin{figure}[t]
 \centering
\begin{tikzpicture}[thick]
\draw [](0,0) -- (2,0) -- (1,1.4) -- cycle;
\draw [<->] (0.7, 0.3) -- (1.3, 0.3);
\node () at (1, 0.55) {$\hat \tau$};
\fill (0,0) circle [radius=1.5pt];
\fill (2,0) circle [radius=1.5pt];
\fill (1,1.4) circle [radius=1.5pt];
\node () at (-1,-0.2) {$a+b\sqrt{-1} = \mu_1$};
\node () at (3,-0.2) {$\mu_2 = a-b\sqrt{-1}$};
\node () at (1,1.65) {$\mu_3 = c$};
\end{tikzpicture}
\caption{The graph associated with a Cartan subalgebra of $\gl_3(\R)$.}\label{gl3}
\end{figure}

\begin{example}\label{exCartan}
The graph corresponding to the Cartan subalgebra
$$\left(\begin{array}{ccc}a & b & 0 \\-b & a & 0 \\0 & 0 & c\end{array}\right) \subset \gl_3(\R)$$
is a triangle with the involution $\hat \tau$ interchanging two vertices (see Figure~\ref{gl3}).
\end{example}
From Theorem~\ref{thm2}, we get the following proposition.
\begin{proposition}\label{zeroRankTypes}
For the argument shift system on any real form $\g$ of $\gl_n(\Complex)$, all rank zero non-degenerate singular points have the same Williamson type $(k_e,k_h,k_f)$ uniquely determined by the graph $(G, \hat \tau)$ associated with the Cartan subalgebra $\mathfrak t$ containing $J$. Namely,
\begin{enumerate}
\item  for real forms of non-compact type $k_e$ is the number of edges $e \in G$ such that $\hat \tau(e) = -e$, and $k_h$ is the number of  of edges $e \in G$ such that $\hat \tau(e) = e$;
\item  for real forms of compact type $k_e$ is the number of edges $e \in G$ such that $\hat \tau(e) = e$, and $k_h$ is the number of  of edges $e \in G$ such that $\hat \tau(e) = -e$;
\item  for any real form, $k_f$ is the number of pairs of edges of $G$ interchanged by $\hat \tau$.
\end{enumerate}

\end{proposition}
\begin{remark}
Note that Theorem~\ref{thm2} can be reformulated in the same way for an arbitrary real polynomial matrix system and an arbitrary nodal spectral curve. In the general case, the graph $G$ should be replaced by the {dual graph} of the spectral curve (see Section~\ref{sec:ncrd} for the definition of this graph).
\end{remark}
\begin{example}
Let $\g \simeq \gl_3(\R)$, and let $\mathfrak t$ be the Cartan subalgebra from Example~\ref{exCartan}. Then rank zero singularities of the corresponding integrable system are of focus-elliptic type.
\end{example}

\begin{corollary}\label{corun} \begin{enumerate}
\item For the argument shift system on $\un_n$, all rank zero non-degenerate singular points are of pure elliptic type.

\item For the argument shift system on $\gl_n(\R)$ with $J$ belonging to a split (i.e., $\R$-diagonalizable) Cartan subalgebra, all rank zero non-degenerate singular points are of pure hyperbolic type.

\item For the argument shift system on $\gl_k(\H)$, all rank zero non-degenerate singular points are of type $(k, 0,2k(k-1))$.
\end{enumerate}
\end{corollary}
\begin{proof}
The first two statements follow from the fact that the involution $\hat \tau$ in these cases is the identity. For the third statement, one uses the fact that all Cartan subalgebras in $\gl_k(\H)$ are mutually conjugate~\cite{sugiura1959conjugate}. The corresponding complete graph $G$ has $2k$ vertices none of which are fixed by $\hat \tau$, so the result follows from Proposition~\ref{zeroRankTypes}.
\end{proof}
\begin{remark}
The first two statements of Corollary~\ref{corun} have been recently obtained by T.\,Ratiu and D.\,Tarama by means of direct computations (see~\cite{ratiu2015u}  for the first statement and \cite{Tarama} for the second statement). 
\par
The interest in singularities for shift of argument systems comes from the fact that these systems can be interpreted as integrable geodesic flows on the corresponding Lie groups, or, in V.\,Arnold's terminology, as generalized rigid bodies. Rank zero singular points described in  Corollary~\ref{corun} correspond to relative equilibria (i.e., steady rotations) of these rigid bodies. The problem of computing their types is of interest, in particular, in view of the classical Euler theorem saying that steady rotations of a free (three-dimensional) rigid body about the short and long principal axes are stable (elliptic), while rotations about the middle axis are unstable (hyperbolic). As follows from Proposition \ref{zeroRankTypes}, for rigid bodies related to polynomial matrix systems, this kind of behavior is not possible: all their zero rank singularities are of the same type. This difference can be explained by the fact that polynomial matrix systems come from the usual loop algebra of $\gl_n$, while the classical rigid body is related to the \textit{twisted} loop algebra. A detailed study of singularities for systems associated with twisted loop algebras (including the free rigid body and its multidimensional generalizations) will be the subject of our forthcoming publication.
\end{remark}
\medskip
\subsection{The Lagrange top}\label{sec:lagrange}
Our second example is the Lagrange top. For a treatment of singularities of this system based on direct computations, see, e.g.,~\cite{oshemkov1991fomenko, lewis1992heavy}.

The Lagrange top is a rigid body with a fixed point placed in a uniform gravity field, such that the axis through the center of mass and the fixed point is also an axis of symmetry of the body. The equations of motion for the Lagrange system read
\begin{align}
\label{Lagrange}
\begin{cases}
\dot M = M \times \Omega+  \Gamma \times J\,, \\
\dot \Gamma \,\,= \Gamma \times \Omega\,,
\end{cases}
\end{align}
where $M$, $\Omega$, $\Gamma$, $J \in \R^3$ are, respectively, the angular momentum, the angular velocity, the unit vector in the direction of gravity, and the constant vector $ (0,0,a)$ joining the center of mass with the fixed point (all vectors are written in body coordinates). The angular momentum and the angular velocity are related by $M = I\Omega$, where $I= \mathrm{diag}(1, 1, b)$ is the inertia tensor. The product in~\eqref{Lagrange} is the cross product of vectors, while the dot stands for the derivative with respect to time. \begin{remark}
Note that the vector $\Gamma$ determines the position of the body only up to rotation about the vertical axis. Thus, system~\eqref{Lagrange} describes the dynamics of the top modulo the rotational symmetry. 
\end{remark}
Equations of the Lagrange top can be written as a single equation
\begin{align*}
\diffXp{t} (\Gamma + M\lambda + J \lambda^2) = (\Gamma + M\lambda + J \lambda^2) \times (\Omega + J\lambda)\,,
\end{align*}
which, being reformulated in terms of commutators of $3 \times 3$ skew-symmetric matrices, coincides with the Lax representation for the Lagrange top found by T.\,Ratiu and P.\,Van Moerbeke~\cite{ratiu1982lagrange}.  
Note that the corresponding Lie algebra $\so_3(\R) \simeq \su_2$ is a real form of $\sL_2(\Complex)$, so the Lagrange system fits directly into the framework of the present paper. However, in order to apply Theorem~\ref{thm2}, one needs to replace the three-dimensional representation of  $\su_2$ with the fundamental representation (cf. Section 2 of~\cite{Gavrilov}). This leads to the Lax representation
$
\dot  L = [L, A]
$, where 
\begin{align*}
L = \left(\begin{array}{cc}\sqrt{-1}(\gamma_3 + m_3 \lambda + a \lambda^2) & \gamma_1 + \sqrt{-1} \gamma_2 + (m_1 + \sqrt{-1}m_2)\lambda \\ -\gamma_1 + \sqrt{-1} \gamma_2 + (-m_1 + \sqrt{-1}m_2)\lambda & -\sqrt{-1}(\gamma_3 + m_3 \lambda + a\lambda^2)\end{array}\right),
\end{align*}
$\gamma_i$ and $m_i$ are the components of the vectors $\Gamma$ and $M$ respectively, and
\begin{align*}
A =\frac{1}{2} \left(\begin{array}{cc}\sqrt{-1}(b^{-1}m_3 + a \lambda) & m_1 + \sqrt{-1}m_2 \\ -m_1 + \sqrt{-1}m_2 & -\sqrt{-1}(b^{-1}m_3 + a \lambda)\end{array}\right).
\end{align*}
A direct computation shows that $A$ can be written in the form $\phi(\lambda, L)_+$, where $\phi$ is given by
\begin{align}\label{phiLagr}
\phi(\lambda, \mu) = \frac{(b-1)\sqrt{-1}}{4ab}(\lambda^{-3}\mu^2 + a^2\lambda) + \frac{1}{2}\lambda^{-1}\mu\,.
\end{align}
Thus, the Lagrange top is precisely the polynomial matrix system on $\pazocal P_2^J (\su_2)$. The equation~\eqref{loopI} corresponding to the function $\phi(\lambda, \mu)$ given by formula~\eqref{phiLagr} describes the motion of the top itself, while other choices of $\phi$ lead to symmetries of the top. In particular, $\phi = \lambda^{-2}\mu$ gives rise to a rotational symmetry with respect to the axis of the body (not to be confused with rotations about the vertical axis which are already modded out).\par
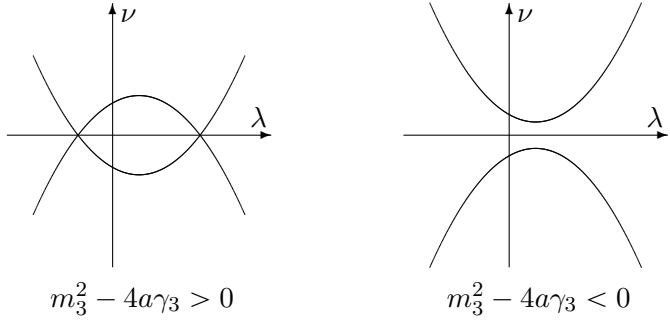
\begin{figure}[t]
{
\centering
{
\begin{picture}(500,120)
\put(110,15){
\qbezier(10,80)(50,-10)(90,80)
\qbezier(10,20)(50,110)(90,20)
\put(0,50){\vector(1,0){100}}
\put(92,53){$\lambda$}
\put(40,0){\vector(0,1){100}}
\put(43,93){$\nu$}
\put(0,-15){\, $\quad m_3^2 - 4a\gamma_3 > 0$}
}
\put(260,15){
\qbezier(10,100)(50,10)(90,100)
\qbezier(10,0)(50,90)(90,0)
\put(0,-15){\, $\quad m_3^2 - 4a\gamma_3 < 0$}
\put(0,50){\vector(1,0){100}}
\put(92,53){$\lambda$}
\put(40,0){\vector(0,1){100}}
\put(43,93){$\nu$}
}
\end{picture}
\caption{Real part of the spectral curve for rank zero singular points of the Lagrange top.}\label{LagrangeCurve}
}
}
\end{figure}
One can use Theorems~\ref{thm1} and~\ref{thm2} to obtain a complete description of singularities for the Lagrange system. Here we only consider the most interesting from the point of view of mechanics rank zero singular points given by diagonal matrices $L$. These points correspond to the situation when both the top and its rotation axis are vertical (this is what kids playing with the top are trying to achieve). The associated spectral curve $C_L$ is
\begin{align}\label{LagrangeSC}
\mu^2 + (\gamma_3 + m_3 \lambda + a \lambda^2)^2 = 0\,.
\end{align}
The antiholomorphic involution $\tau$ on $C_L$ is given by $(\lambda , \mu) \mapsto (\bar \lambda, -\bar \mu)$, so it is convenient to set $\nu := \sqrt{-1} \mu$. Then the equation of the spectral curve becomes
$$
\nu^2 = (\gamma_3 + m_3 \lambda + a \lambda^2)^2
$$
with the standard real structure $(\lambda , \nu) \mapsto (\bar \lambda, \bar \nu)$. This curve is nodal provided that the discriminant $D = m_3^2 - 4a\gamma_3$ is not zero. When $D > 0$, the spectral curve has two self-intersection points (see Figure~\ref{LagrangeCurve}), so by Theorem~\ref{thm2} the corresponding singular point has elliptic-elliptic type (note that all double points of the spectral curve are essential since $L$ is a diagonal matrix). When, on the contrary, $D< 0$, both double points are non-real, and we have a focus-focus singularity (cf. Proposition \ref{zeroRankTypes} saying that in the degree one case the types of all rank zero points are the same). Thus, we have recovered a classical result that the rotation of a top about the vertical axis is stable provided that the angular momentum is large enough ($m_3^2 > 4a\gamma_3$).
\begin{remark}
When $D= 0$, the spectral curve has an $A_3$ singularity. The corresponding bifurcation under which an elliptic-elliptic point becomes a focus-focus point is known as the (supercritical) \textit{Hamiltonian Hopf bifurcation}. We conjecture that $A_3$ singularities on the spectral curve always lead to this type of bifurcations.
\end{remark}

\bigskip

\bibliographystyle{plain}
\bibliography{SingCurves}
\end{document}